\documentclass[lettersize,journal]{IEEEtran}
\usepackage{amsmath,amsfonts}
\usepackage{array}
\usepackage[caption=false,font=normalsize,labelfont=sf,textfont=sf]{subfig}
\usepackage{textcomp}
\usepackage{stfloats}
\usepackage{url}
\usepackage{verbatim}
\usepackage{graphicx}
\usepackage{cite}
\usepackage{bm}
\usepackage[ruled]{algorithm2e}
\usepackage{tikz,xcolor} 
\usepackage{booktabs}
\usepackage{multirow}%
\usepackage{footmisc}
\usepackage{hyperref}
\usepackage{amsthm}
\UseRawInputEncoding
\newtheorem{definition}{Definition}%

\newtheorem{Lemma}{Lemma}%

\begin{document}

\definecolor{lime}{HTML}{A6CE39}
\DeclareRobustCommand{\orcidicon}{
\begin{tikzpicture}
\draw[lime, fill=lime] (0,0)
circle[radius=0.13]
node[white]{{\fontfamily{qag}\selectfont \tiny \.{I}D}};
\end{tikzpicture}
\hspace{-2mm}
}
\foreach \x in {A, ..., Z}{%
\expandafter\xdef\csname orcid\x\endcsname{\noexpand\href{https://orcid.org/\csname orcidauthor\x\endcsname}{\noexpand\orcidicon}}
}

\title{Quantum Key-Recovery Attacks on FBC Algorithm}
\author{Yan-Ying Zhu, Bin-Bin Cai, Fei Gao, and Song Lin
\thanks{$\bullet$ Yan-Ying Zhu is with the College of Computer and Cyber Security, Fujian Normal University, Fuzhou, 350117, China. }
\thanks{$\bullet$ Bin-Bin Cai is with the College of Computer and Cyber Security,
Fujian Normal University, Fuzhou, 350117, China, the State Key Laboratory
of Networking and Switching Technology, Beijing University of Posts and
Telecommunications, Beijing, 100876, China, and with the Digital Fujian
Internet-of-Things Laboratory of Environmental Monitoring, Fujian Normal
University, Fuzhou, 350117, China (e-mail: cbb@fjnu.edu.cn).}
\thanks{$\bullet$ Fei Gao is with the State Key Laboratory of Networking and Switching Technology, Beijing University of Posts and Telecommunications, Beijing, 100876, China. }
\thanks{$\bullet$ Song Lin is with the College of Computer and Cyber Security, Fujian Normal University, Fuzhou, 350117, China (e-mail: lins95@fjnu.edu.cn).}
\thanks{(Corresponding authors: Bin-Bin Cai, Song Lin.)}
}

\markboth{Journal of \LaTeX\ Class Files,~Vol.~XX, No.~X, August~2025}%
{Shell \MakeLowercase{\textit{et al.}}: A Sample Article Using IEEEtran.cls for IEEE Journals}


\maketitle

\begin{abstract}
With the advancement of quantum computing, symmetric cryptography faces new challenges from quantum attacks.  These attacks are typically classified into two models: Q1 (classical queries) and Q2 (quantum superposition queries). In this context, we present a comprehensive security analysis of the FBC algorithm considering quantum adversaries with different query capabilities. In the Q2 model, we first design 4-round polynomial-time quantum distinguishers for FBC-F and FBC-KF structures, and then perform $r(r>6)$-round quantum key-recovery attacks. Our attacks require $O(2^{(2n(r-6)+3n)/2})$ quantum queries, reducing the time complexity by a factor of $2^{4.5n}$ compared with quantum brute-force search, where $n$ denotes the subkey length. Moreover, we give a new 6-round polynomial-time quantum distinguisher for FBC-FK structure. Based on this, we construct an $r(r>6)$-round quantum key-recovery attack with complexity $O(2^{n(r-6)})$. Considering an adversary with classical queries and quantum computing capabilities, we demonstrate low-data quantum key-recovery attacks on FBC-KF/FK structures in the Q1 model. These attacks require only a constant number of plaintext-ciphertext pairs, then use the Grover algorithm to search the intermediate states, thereby recovering all keys in $O(2^{n/2})$ time.
\end{abstract}

\begin{IEEEkeywords}
FBC algorithm, quantum attacks, key-recovery, low-data
\end{IEEEkeywords}

\section{\label{sec:1}Introduction}
\IEEEPARstart{I}{n} recent years, the advent of quantum algorithms has raised growing concerns about potential security threats \cite{Su2025, Wu2024,Wu2025}. With the continuous advancement of quantum computing technology, many cryptographic schemes considered secure in the classical setting are facing unprecedented challenges \cite{Kuwakado2010, Kuwakado20102,Kaplan2016}. Indeed, the security of traditional cryptosystems is founded on specific computationally hard problems, whose security assumptions have long been defined under the classical attacker model. However, once the adversary possesses quantum capabilities, the original assumptions of the cryptographic security system no longer hold, and its security boundaries are broken. Since Shor \cite{Shor1994,Shor1997} proposed the famous polynomial-time algorithm, asymmetric encryption algorithms (e.g., RSA \cite{RSA1978}) that rely on integer factorization and discrete logarithm problem have been proven to be vulnerable to quantum attacks. For unstructured search problems, Grover algorithm \cite{Grover1996} achieves a quadratic speedup compared with classical algorithm. When applied to key-recovery attacks on symmetric cryptography, it effectively reduces the key length by half. In symmetric cryptography, the security of most schemes is not based on structured mathematical problems. Without known vulnerabilities, quantum attacks on symmetric cryptography mainly depend on Grover algorithm for brute-force key search. Consequently, doubling the key length is generally sufficient to achieve quantum resistance, suggesting that symmetric schemes are less affected by quantum threats. Nevertheless, Kuwakado and Morii \cite{Kuwakado20102} introduced a novel quantum attack method against certain generic constructions of symmetric key schemes. They showed that the Even-Mansour construction can be broken in polynomial-time $O(n)$ using Simon algorithm \cite{Simon1997}, where $n$ denotes the key length. This finding reveals that Grover algorithm is not the only threat to symmetric cryptography in the quantum setting. The following work by Kaplan et al. \cite{Kaplan2016} showed that several widely used modes of operation for authentication and authenticated encryption (such as CBC-MAC, PMAC, GMAC, and some CAESAR competition candidates) were also broken by the Simon algorithm. These attacks achieve exponential speedups compared with classical attacks on cryptographic schemes. Subsequently, more researchers began using Simon algorithm to design quantum attacks on block cipher structures \cite{Luo2019,You2020,Qian2021,Li2022,Li2021}.

Besides, Simon algorithm can also be combined with Grover algorithm to attack certain cipher structures with expanded keys. In 2017, Leander et al. \cite{Leander2017} proposed the
Grover-meets-Simon algorithm, in which Simon algorithm is used as a distinguisher within Grover search. Specifically, the adversary first uses Grover algorithm to search for partial keys. If the guessed keys are correct, the adversary can obtain a periodic function, and then use Simon algorithm to determine its period. For the key-recovery problem of the FX structure \cite{Kilian1996}, the query complexity of the Grover search algorithm is \(O((n+l)2^{(n+l)/2})\), while the Grover-meets-Simon algorithm requires only \(O((n+l)2^{n/2})\) queries, where $n$ and $l$ denote the original key length and the key expansion length, respectively. In addition, for encryption algorithms that employ whitening key expansion techniques \cite{Onions2001}, the Grover-meets-Simon algorithm can also significantly reduce the query complexity of key-recovery attacks. Building on this technique, Dong et al. \cite{Dong2018} proposed new quantum key-recovery attacks by adding several rounds to the quantum distinguisher on Feistel construction. Since then, the Grover-meets-Simon algorithm has been widely applied to Feistel structures \cite{Ito20192,Canale2022}, Type-1/2/3 generalized Feistel structures(GFSs) \cite{Dong2019,Ni2019,Sun2023,S2020,Zhang2023}, and Feistel variants \cite{Cui2021,Bo2021,Kun2022}. The emergence of these quantum attacks has also promoted the development of various quantum algorithms for analyzing cryptographic protocols \cite{Bonnetain2021,Zhou2021,Bon20192,Tan2022,Zhou2023}.

Obviously, the above results are obtained in the Q2 model \cite{Kaplan20162} (i.e., adversaries with the abilities of quantum superposition queries and quantum computation). Considering the actual capabilities of the adversary, quantum attacks in the Q1 model \cite{Kaplan20162} (i.e., adversaries with the abilities of classical queries and quantum computation) have also been proposed. In 2019, Bonnetain et al. \cite{Bonnetain2019} proposed the offline Simon algorithm based on the Grover-meets-Simon algorithm. The algorithm allows quantum attacks to be performed even when the encryption oracle accepts only classical queries, and its complexity is comparable to that of the Grover-meets-Simon algorithm. In 2020, Rahman et al. \cite{Rahman2020} proposed a new attack on hash counters using an offline computation approach, building on the work of Ref. \cite{Bonnetain2019}. In the same year, Bonnetain et al. \cite{Bonnet2020} provided the first complete implementation of the offline Simon algorithm and evaluated the attack costs on several lightweight cryptographic schemes. Subsequently, Yu et al. \cite{Yu2023} analyzed the details of the offline Simon algorithm. In 2022, Daiza et al. \cite{Daiza2022} proposed a new low-data quantum key-recovery attack method on 3-round Feistel-KF structure in the Q1 model. Recently, Xu et al. \cite{Xu2024} designed low-data key-recovery attacks for different block cipher structures, including Lai-Massey, Misty, Type-1 GFS, SMS4 and MARS. These algorithms have stronger practicality in attacking block ciphers.

\begin{table}[t]
\caption{Results of quantum cryptanalysis on FBC algorithm}\label{table1}
\begin{tabular}{  >{\centering\arraybackslash}p{1.8cm}
  >{\centering\arraybackslash}p{1.6cm}
  >{\centering\arraybackslash}p{1.7cm}
  >{\centering\arraybackslash}p{2cm}}
\toprule
Structure & Distinguisher & Key-recovery rounds & Complexity (log) \\
\midrule
\multirow{2}{*}{FBC-F \cite{Bo2021}} & \multirow{2}{*}{3 (Q2)} & 5 & $\frac{3n}{2}$\\
\cmidrule{3-4}
& & $r>5$ & $\frac{3n+(r-5)\cdot 2n}{2}$\\
\midrule
\multirow{2}{*}{FBC-F \cite{Kun2022}} & \multirow{2}{*}{4 (Q2)} & 6 & $\frac{3n}{2}$\\
\cmidrule{3-4}
& & $r>6$ & $\frac{3n+(r-6)\cdot 2n}{2}$\\
\midrule
\multirow{2}{*}{FBC-F (ours)} & \multirow{2}{*}{4 (Q2)} & 6 & $\frac{3n}{2}$\\
\cmidrule{3-4}
& & $r>6$ & $\frac{3n+(r-6)\cdot 2n}{2}$\\
\midrule
\multirow{2}{*}{FBC-KF (ours)} & 4 (Q2) & $r>6$ & $\frac{3n+(r-6)\cdot 2n}{2}$\\
\cmidrule{2-4}
&  $-$ & 4 (Q1) & $\frac{n}{2}$\\
\midrule
\multirow{2}{*}{FBC-FK (ours)} & 6 (Q2) & $r>6$ & $(r-6)\cdot n$\\
\cmidrule{2-4}
&  $-$ & 5 (Q1) & $\frac{n}{2}$\\
\bottomrule
\end{tabular}
\end{table}

\noindent\textbf{Contribution} In this paper, based on the different query capabilities of adversaries, we propose key-recovery attacks for the FBC algorithm in the Q1 and Q2 models, respectively. First, we analyze the algebraic structures of the FBC-F, FBC-KF, and FBC-FK schemes. In the Q2 model, we construct 4-round polynomial-time quantum distinguishers on FBC-F and FBC-KF structures, and design $r(r>6)$-round quantum key-recovery attacks. Our attacks require $O(2^{(2n(r-6)+3n)/2})$ quantum queries, thereby reducing the time complexity by a factor of $2^{4.5n}$ compared with quantum brute-force search, where $n$ denotes the subkey length. Moreover, we construct a new 6-round polynomial-time quantum distinguisher on FBC-FK structure in quantum chosen-plaintext attack (qCPA) setting. Based on this distinguisher, a quantum key-recovery attack can be performed with a time complexity of $O(2^{n(r-6)})$. The comparison of previous quantum attacks and our attacks on FBC structure is shown in Table \ref{table1}. For the FBC-F structure, the number of rounds covered by our constructed quantum distinguisher is of the same order of magnitude as that in Ref. \cite{Kun2022}, and both outperform the results presented in Ref.~\cite{Bo2021}. It is worth noting that our construction method differs from that of Ref. \cite{Kun2022}, and the resulting periodic function exhibits distinct structural characteristics. This indicates that the weakness of the cipher is not confined to a specific configuration. Building on this approach, we further propose a 4-round quantum distinguisher for the FBC-KF structure and a 6-round quantum distinguisher for the FBC-FK structure. Moreover, we design corresponding quantum key-recovery attacks for the $r$-round versions of these structures.

Second, in the Q1 model, we present low-data quantum key-recovery attacks targeting both the 4-round FBC-KF structure and the 5-round FBC-FK structure. By leveraging Grover algorithm to efficiently search for intermediate states, the time complexity for each Grover search step is $O(2^{n/2})$. As a result, all keys can be recovered with only $O(1)$ data complexity and $O(2^{n/2})$ time complexity, while the classical storage complexity remains negligible.

\noindent\textbf{Organization} This paper is organized as follows. Sect. \ref{sec2} introduces some essential preliminaries. In Sect. \ref{sec3}, we propose new quantum distinguishers for the FBC-F, FBC-KF, and FBC-FK structures, along with corresponding $r$-round key-recovery attacks. Sect.~\ref{sec4} presents new low-data key-recovery attacks on FBC-KF/FK structures in the Q1 model. Finally, the conclusions are presented in Sect. \ref{sec5}.

\section{Preliminaries}\label{sec2}
In this section, based on the method of key injection, we introduce three structural variants of the FBC algorithm, namely FBC-F, FBC-KF, and FBC-FK structures. We then give several quantum algorithms and the low-data quantum key-recovery attack.
\subsection{The construction settings of the FBC algorithm}

    \begin{figure*}[]
        \centering
        \subfloat[\footnotesize Feistel-F]{
            \includegraphics[height=0.17\textwidth,width=0.17\textwidth]{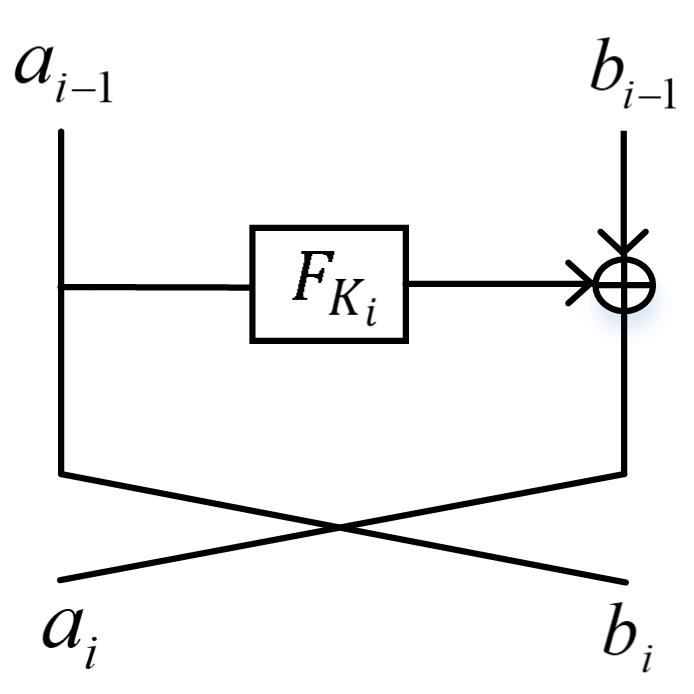}
        \label{fig1(a)}
        }
        \hfil
        \subfloat[\footnotesize Feistel-KF]{
            \includegraphics[height=0.17\textwidth,width=0.2\textwidth]{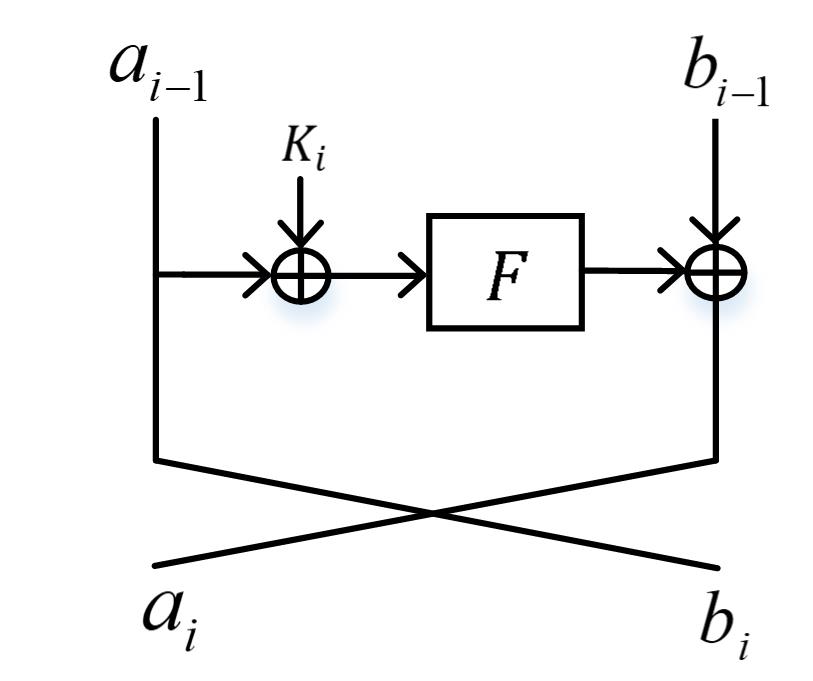}
        \label{fig1(b)}
        }
        \hfil
        \subfloat[\footnotesize Feistel-FK]{
            \includegraphics[height=0.17\textwidth,width=0.198\textwidth]{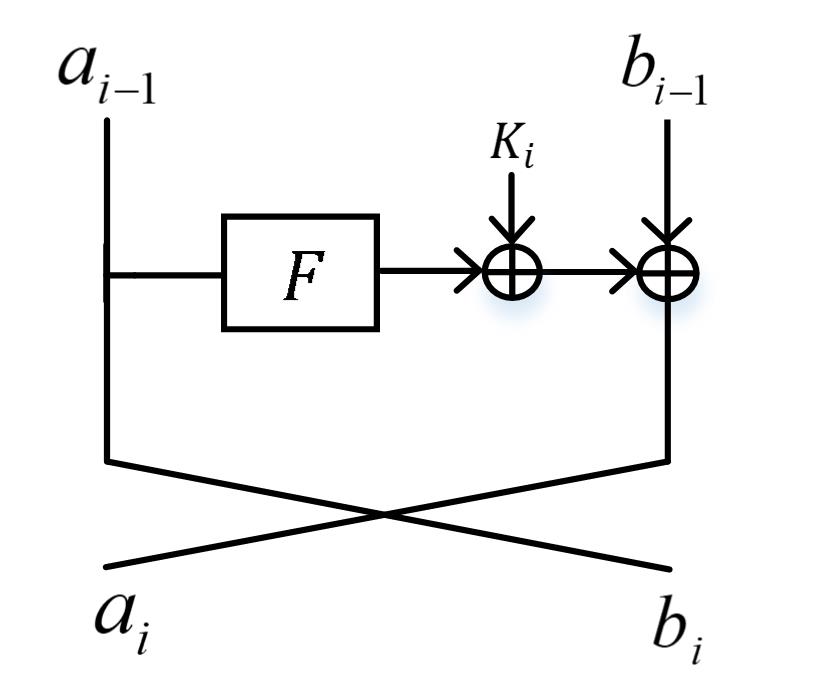}
        \label{fig1(c)}
        }
        \caption{Feistel-F/KF/FK structures.}
        \label{fig1}
    \end{figure*}

The traditional $r$-round Feistel structure (also known as Feistel-F \cite{Ito20192}) takes a plaintext $P=(a_{0}, b_{0})$ as input, where $a_{0}, b_{0} \in\{0,1\}^{n / 2}$, and its $i$-th round structure is illustrated in Fig. \ref{fig1(a)}. Then, it calculates $(a_{i}, b_{i})=(b_{i-1} \oplus F_{k_{i}}(a_{i-1}), a_{i-1})$, where $F_{K_{i}}:\{0,1\}^{n / 2} \rightarrow\{0,1\}^{n / 2}$ is a random function with a subkey $K_{i}\in \{0,1\}^{n/2}$. However, instantiating a distinct key-dependent random function for each round incurs high implementation cost, which is often practically infeasible. In order to design block ciphers suitable for practical applications, Ito et al. \cite{Ito20192} proposed the Feistel-KF structure (the $i$-th round is shown in Fig. \ref{fig1(b)}).
In this design, the subkey \(K_i \) is XORed into \(a_{i-1}\), and the result is used as the input of the public function \(F\). The output of $i$-th round is \(\left(a_i, b_i\right) = \left(b_{i-1} \oplus F\left(a_{i-1} \oplus K_i\right), a_{i-1}\right)\). If the subkey \(K_i\) is injected outside the function \(F\), the output becomes \(\left(a_i, b_i\right) = \left(b_{i-1} \oplus F\left(a_{i-1}\right) \oplus K_i, a_{i-1}\right)\). The \(i\)-th round of this variant is shown in Fig. \ref{fig1(c)}. It is called the Feistel-FK structure.

The FBC algorithm \cite{KFeng2019} adopts a $4$-branched dual Feistel structure as shown in Fig. \ref{fig2}. The input length is $4n$ bits in total and is divided into four branches, each consisting of $n$ bits. Similar to the naming convention used by Ito et al. \cite{Ito20192}, we name the FBC algorithm as the FBC-F structure. $F_{1}^{i}$ and $F_{2}^{i}$ are the $i$-th round of random functions which absorb the independent round subkeys $k_{1}^{i}$ and $k_{2}^{i}$, respectively. Assume that $(x_{0}^{i-1}, x_{1}^{i-1}, x_{2}^{i-1}, x_{3}^{i-1})\in(\{0,1\}^{n})^{4}$ is the input of the $i$-th round FBC-F structure, the corresponding output is $(x_{0}^{i}, x_{1}^{i}, x_{2}^{i}, x_{3}^{i})=(x_{1}^{i-1} \oplus F_{1}^{i}(x_{0}^{i-1}), x_{0}^{i-1} \oplus x_{2}^{i-1} \oplus F_{2}^{i}(x_{3}^{i-1}), x_{3}^{i-1} \oplus x_{1}^{i-1} \oplus F_{1}^{i}(x_{0}^{i-1}), x_{2}^{i-1} \oplus F_{2}^{i}(x_{3}^{i-1}))$.

Based on the key injection method, we extend FBC-F structure into two new variants named FBC-KF and FBC-FK structures, as shown in Fig. \ref{fig3} and Fig. \ref{fig4}, respectively. In the FBC-KF structure, the output of the $i$-th round is
$(x_{0}^{i}, x_{1}^{i}, x_{2}^{i}, x_{3}^{i})=(x_{1}^{i-1} \oplus F_{1}^{i}(x_{0}^{i-1} \oplus k_{1}^{i}), x_{0}^{i-1} \oplus x_{2}^{i-1} \oplus F_{2}^{i}(x_{3}^{i-1}\oplus k_{2}^{i}), x_{3}^{i-1} \oplus x_{1}^{i-1} \oplus F_{1}^{i}(x_{0}^{i-1}\oplus k_{1}^{i}), x_{2}^{i-1} \oplus F_{2}^{i}(x_{3}^{i-1}\oplus k_{2}^{i}))$.
Similarly, we can obtain the output of the FBC-FK structure as $(x_{0}^{i}, x_{1}^{i}, x_{2}^{i}, x_{3}^{i})=(x_{1}^{i-1} \oplus k_{1}^{i} \oplus F_{1}^{i}(x_{0}^{i-1}), x_{0}^{i-1} \oplus x_{2}^{i-1}\oplus k_{2}^{i} \oplus F_{2}^{i}(x_{3}^{i-1}), x_{3}^{i-1} \oplus x_{1}^{i-1}\oplus k_{1}^{i} \oplus F_{1}^{i}(x_{0}^{i-1}), x_{2}^{i-1}\oplus k_{2}^{i} \oplus F_{2}^{i}(x_{3}^{i-1}))$. Here, the $F_1^i$ and $F_2^i$ functions in Fig. \ref{fig3} and Fig. \ref{fig4} are public functions.

\begin{figure}
    \centering
    \includegraphics[width=0.7\linewidth]{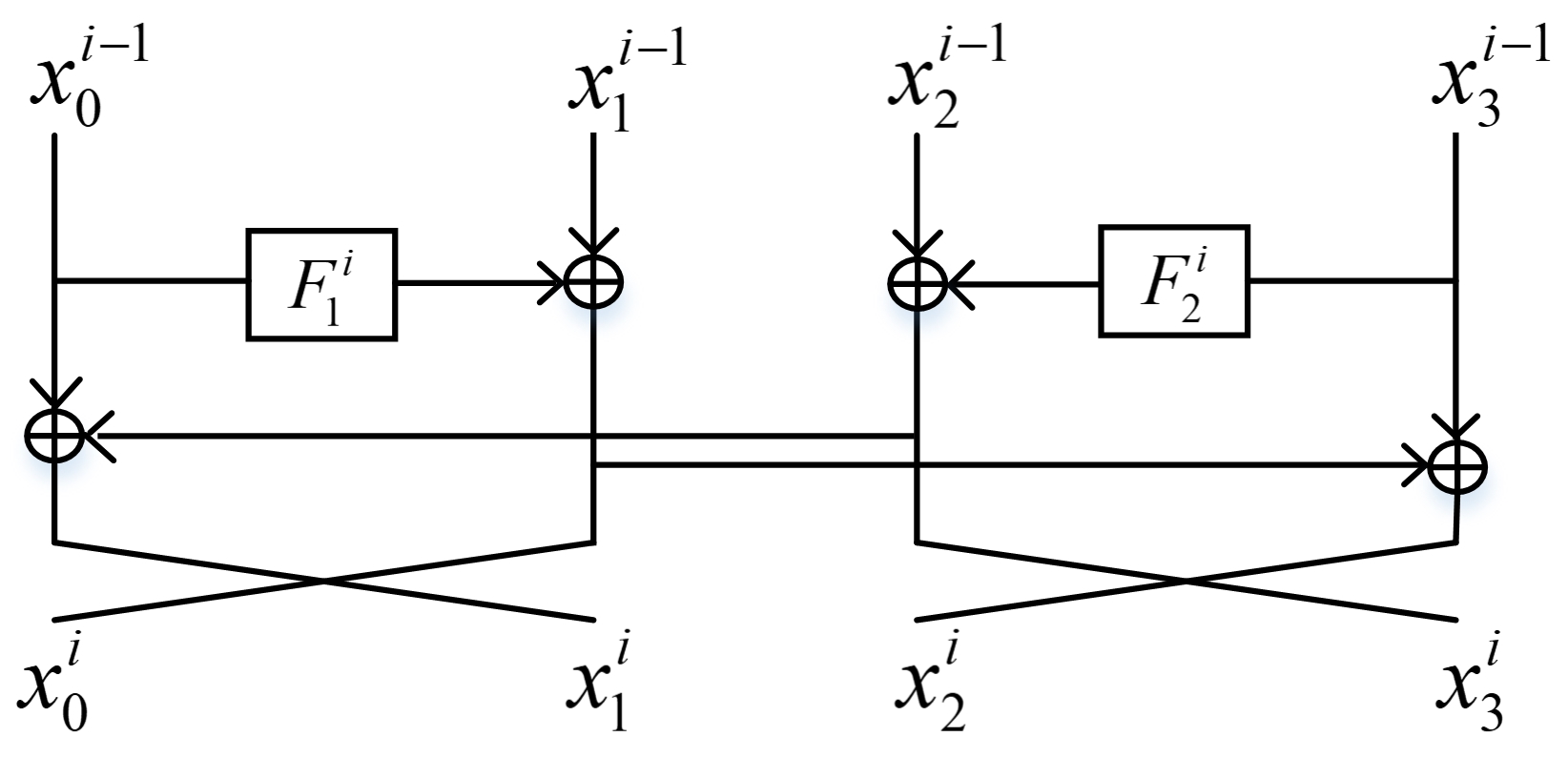}
    \caption{The $i$-th round FBC-F structure.}
    \label{fig2}
\end{figure}
\begin{figure}
    \centering
    \includegraphics[width=0.76\linewidth]{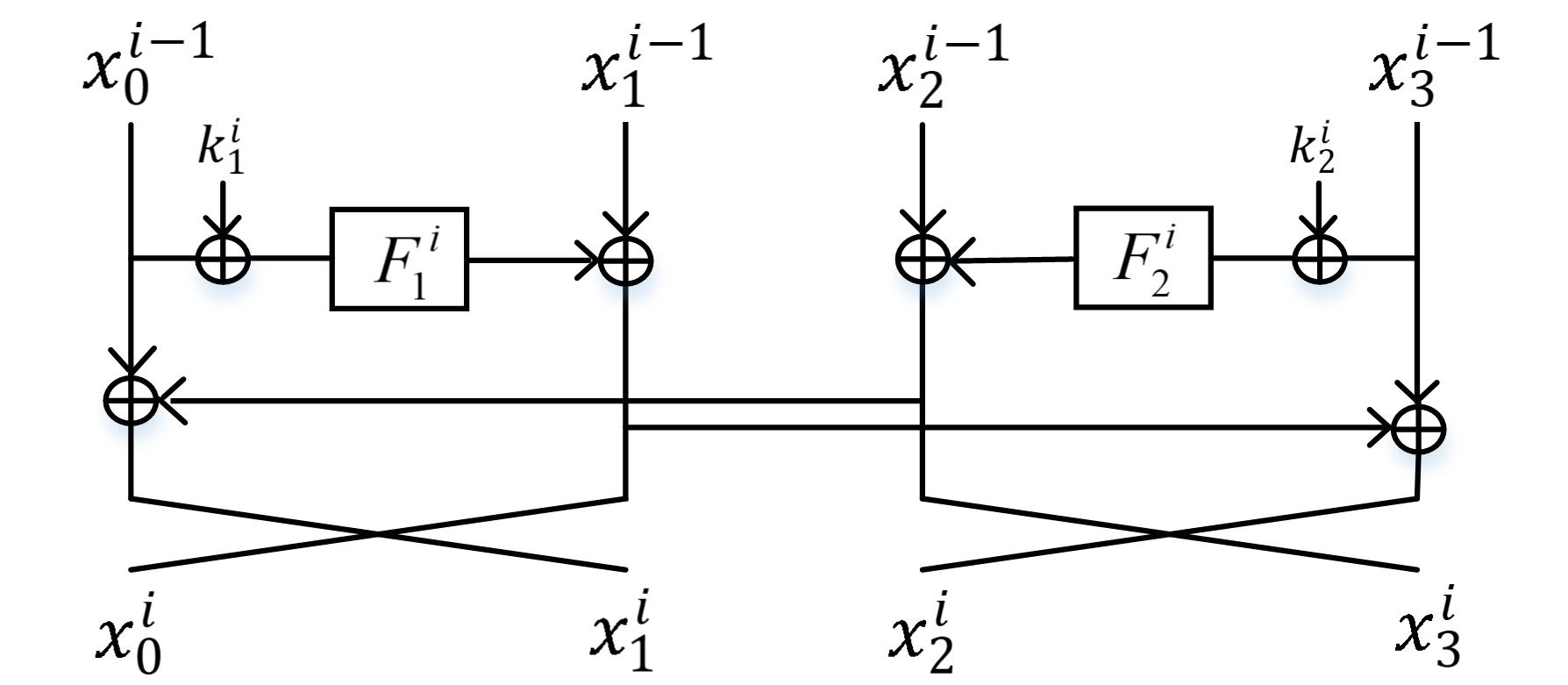}
    \caption{The $i$-th round FBC-KF structure.}
    \label{fig3}
\end{figure}
\begin{figure}
    \centering
    \includegraphics[width=0.7\linewidth]{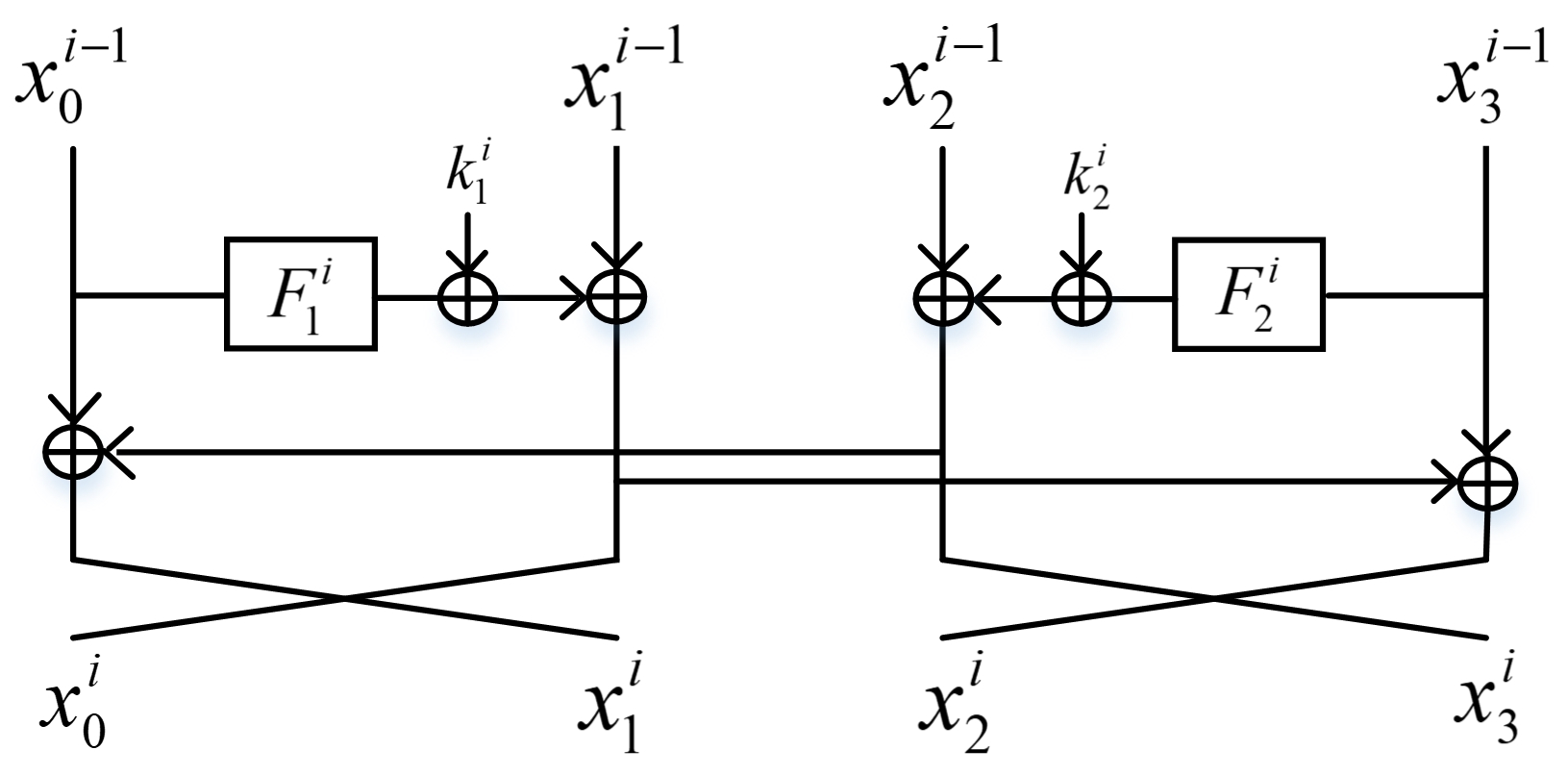}
    \caption{The $i$-th round FBC-FK structure.}
    \label{fig4}
\end{figure}

\par \subsection{Simon algorithm}

Simon algorithm is the first quantum algorithm to achieve exponential speedup over the optimal classical algorithms for solving Simon problem. It can efficiently determine the hidden period of a two-to-one function.
\begin{definition}
(Simon \cite{Simon1997}) Given a function $f:\{0,1\}^{n} \rightarrow\{0,1\}^{n}$, suppose that there exists a unique non-zero period $s \in\{0,1\}^{n}$ such that $f(x)=f(y) \Leftrightarrow x \oplus y \in\{0, s\}$ holds for any $x, y \in\{0,1\}^{n}$. The goal is to find the period $s$.
\end{definition}

Classically, the optimal time complexity to solve this problem is $ O(2^{n / 2})$. However, Simon algorithm only requires $O(n)$ quantum queries to find this hidden period $s$. The algorithm includes the following steps.

Step 1. Initializing two $n$-bit quantum registers to state $|0\rangle^{\otimes n}|0\rangle^{\otimes n}$. Then applying the Hadamard transform $H^{\otimes n}$ to the first register to generate a uniform superposition \[ \left|\psi_{1}\right\rangle=\frac{1}{\sqrt{2^{n}}} \sum_{x \in\{0,1\}^{n}}|x\rangle|0\rangle. \]

Step 2. A quantum oracle query on the function $f$ maps $\left|\psi_{1}\right\rangle$ to
\[
\left|\psi_{2}\right\rangle=\frac{1}{\sqrt{2^{n}}} \sum_{x \in\{0,1\}^{n}}|x\rangle|f(x)\rangle.
\]
Step 3. Measuring the second register yields $f\left(x_{0}\right)$, the first register collapses to
\[
\left|\psi_{3}\right\rangle=\frac{1}{\sqrt{2}}\left(\left|x_{0}\right\rangle+\left|x_{0} \oplus s\right\rangle\right).
\]

Step 4. Applying the Hadamard transform $H^{\otimes n}$ to the first register, the state is now
\[
\left|\psi_{4}\right\rangle=\frac{1}{\sqrt{2^{n+1}}} \sum_{y \in\{0,1\}^{n}}\left((-1)^{y \cdot x_{0}}+(-1)^{y \cdot\left(x_{0} \oplus s\right)}\right)|y\rangle,
\]
where $(-1)^{y \cdot\left(x_{0} \oplus s\right)}=(-1)^{y \cdot x_{0}}(-1)^{y \cdot s}$. The coefficient is non-zero if and only if $y \cdot s=0$. Here,
$y \cdot s=0$ means that $y_{0} \cdot s_{0}+y_{1} \cdot s_{1}+ \cdots +y_{n} \cdot s_{n}=0 \pmod{2}$.

Step 5. Measuring the first register yields a value $y$ that satisfies $y \cdot s=0$.

Repeating the above steps 1-5 $O(n)$ times to obtain $(n-1)$ linearly independent equations, the unique solution $s$ can be found by Gaussian elimination.

\par \subsection{Grover algorithm}

Grover algorithm provides a quadratic speedup for unstructured search problems. It is widely applicable to a variety of problems, including preimage attacks \cite{WangP2020} on cryptographic functions, satisfiability problems \cite{WangR2021}, and database search \cite{Kain2021}.

\begin{definition}
(Grover \cite{Grover1996}) Consider an unstructured database of size $N=2^{n}$. Assume that there is exactly one marked item $x_{0}$, the goal is to find $x_{0}$. In other words, let $f:\{0,1\}^{n} \rightarrow\{0,1\}$ be a Boolean function such that $f(x)=1$ if and only if $x=x_{0}$, otherwise $f(x)=0$. The objective is to identify $x_{0}$.
\end{definition}

In the classical setting, the time complexity to find the target $x_{0}$ is usually $O(2^n)$. However, Grover algorithm can find $x_{0}$ in $O(2^{n/2})$ time in the quantum setting. The steps of Grover algorithm are as follows.

Step 1. Initializing the quantum state $|0\rangle^{\otimes n}|0\rangle$. Applying the Hadamard transform $H^{\otimes n}$ to the first $n$-qubit and performing the operation $HX$ on the last qubit. This results in the state
\[
|\psi\rangle=\frac{1}{\sqrt{2^{n}}} \sum_{x \in\{0,1\}^{n}}|x\rangle|-\rangle,
\]
where the matrix forms of $H$ and $X$ are $H = \frac{1}{\sqrt{2}} \begin{bmatrix} 1 & 1 \\ 1 & -1 \end{bmatrix}$ and $X = \begin{bmatrix} 0 & 1 \\ 1 & 0 \end{bmatrix}$, respectively.

Step 2. Constructing a quantum oracle $\mathcal{O}: \mathcal{O}|x\rangle|-\rangle=(-1)^{f(x)}|x\rangle|-\rangle$, where $f(x)=1$ if $x=x_{0}$. Otherwise, $f(x)=0$.

Step 3. Applying Grover iteration for $R \approx\left\lceil\frac{\pi}{4} \sqrt{2^{n}}\right\rceil$ times to yield
\[
[(2|\psi\rangle\langle\psi|-I) \mathcal{O}]^{R}|\psi\rangle \approx\left|x_{0}\right\rangle|-\rangle.
\]

Step 4. Measuring the first $n$-qubit to get $x_{0}$.

\par \subsection{Grover-meets-Simon algorithm}

In 2017, Leander et al. \cite{Leander2017} proposed the Grover-meets-Simon algorithm to attack the FX structure. They showed that using whitening keys does not significantly enhance security.

\begin{definition}
(Leander \cite{Leander2017}) Suppose that $f:\{0,1\}^{m} \times\{0,1\}^{n} \rightarrow\{0,1\}^{i}$ is a function, where $m$ is in $O(n)$. There exists a unique $i_{0} \in\{0,1\}^{m}$ satisfying $f(i_{0}, x)=f(i_{0}, x \oplus s)$ for any $x \in\{0,1\}^{n}$, where $s \in\{0,1\}^{n}\setminus \{0\}^{n}$. The goal is to find the unique $i_{0}$ and the period $s$.
\end{definition}

\begin{figure}
    \centering
    \includegraphics[width=0.65\linewidth]{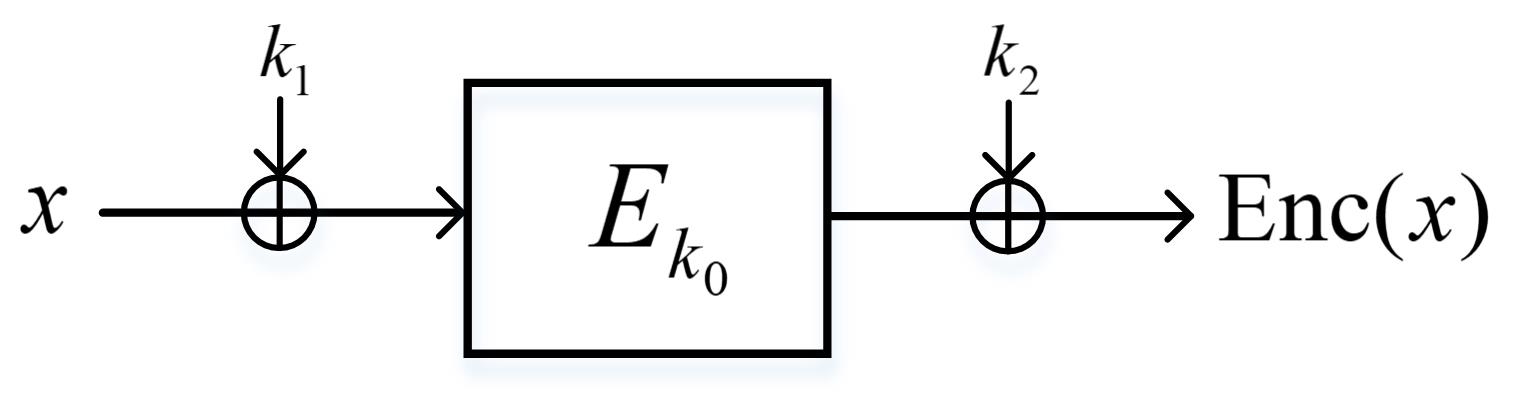}
    \caption{FX structure.}
    \label{fig5}
\end{figure}

The FX structure satisfies $\operatorname{Enc}(x)=E_{k_{0}}(x \oplus k_{1}) \oplus k_{2}$ as shown in Fig. \ref{fig5}. They constructed a function $f(k,x)=E_{k_{0}}(x \oplus k_{1}) \oplus k_{2} \oplus E_{k}(x)$ and used Grover algorithm to exhaustive search for all possible keys $k$ in Ref. \cite{Leander2017}. Obviously, the function $f(k,x)$ is periodic with the period $k_{1}$ for all $x$ when $k=k_{0}$. Otherwise, $f(k,x)$ is not periodic. In this case, the correct $k_{0}$ can be obtained. The period $k_{1}$ can be obtained using Simon algorithm, and then $k_{2}$ can be computed through a simple XOR operation.

\par \subsection{Low-data quantum key-recovery attack}

Daiza et al. \cite{Daiza2022} showed a quantum key-recovery attack on 3-round Feistel-KF structure (see Fig. \ref{fig6}) in the Q1 model. Under the qCPA setting, the adversary uses Grover algorithm to recover all keys in $O(2^{n/2})$ time, where $n$ denotes the key length. The steps of the attack are as follows.
\begin{figure}
    \centering
    \includegraphics[width=0.8\linewidth]{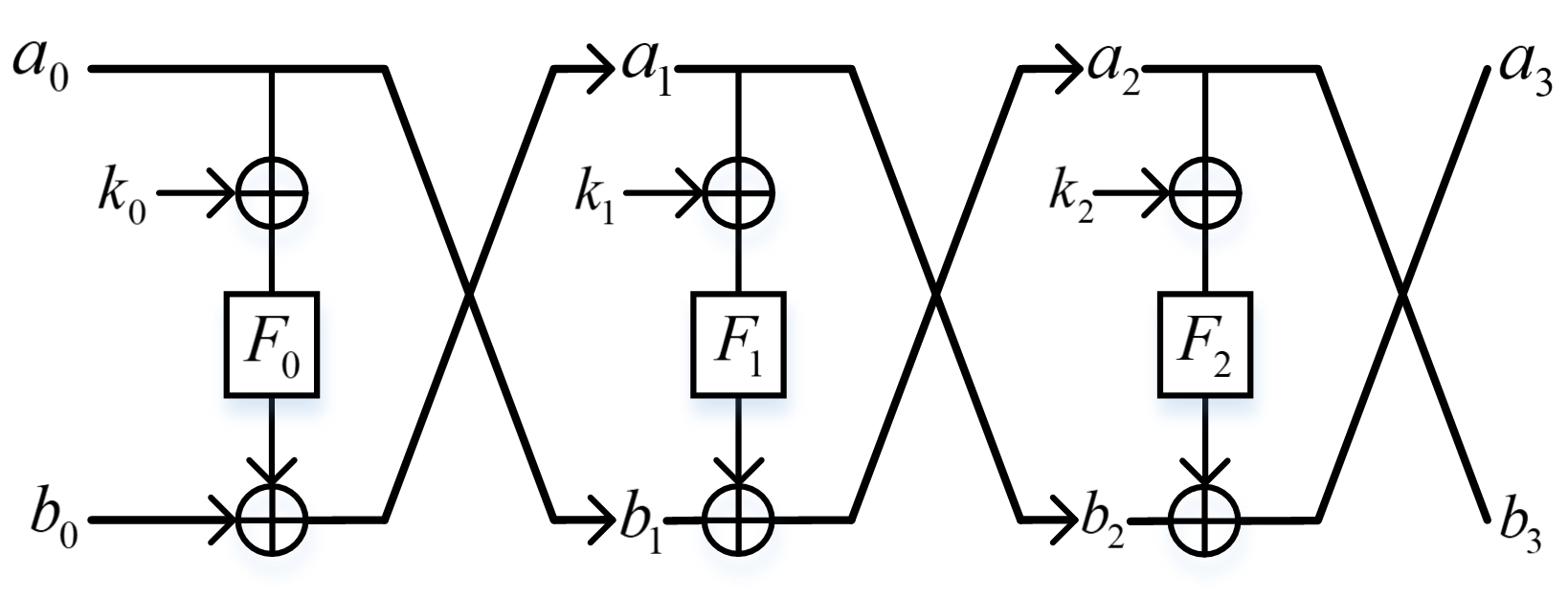}
    \caption{3-round Feistel-KF structure.}
    \label{fig6}
\end{figure}

Based on the 3-round Feistel-KF structure, it is easy to get
\[
\begin{aligned}
b_{3}&=F_{1}\left(k_{1} \oplus F_{0}\left(k_{0} \oplus a_{0}\right) \oplus b_{0}\right) \oplus a_{0},\\
a_{3}&=F_{0}(k_{0} \oplus a_{0}) \oplus b_{0} \oplus F_{2}(k_{2} \oplus F_{1}(k_{1} \oplus F_{0}(k_{0}\\  &\quad \oplus a_{0})\oplus b_{0}) \oplus a_{0}).
\end{aligned}
\]

First, querying the plaintext $(a_{0}, b_{0})=(0^{n / 2}, 0^{n / 2})$ to the encryption oracle to receive $(a_{3}, b_{3})$, where $b_{3}=F_{1}(k_{1} \oplus F_{0}(k_{0}))$. Assume $\beta_{1}=k_{1} \oplus F_{0}(k_{0})$, then $b_{3}=F_{1}(\beta_{1})$, and the value of $\beta_{1}$ can be obtained through Grover algorithm in $O(2^{n/2})$ time.

Then, querying $(a_{0}, b_{0})=(0 \cdots 01, \beta_{1})$ to the encryption oracle to get $(a_{3}, b_{3})$, where $
b_{3}=F_{1}(F_{0}(k_{0} \oplus 0 \cdots 01) \oplus F_{0}(k_{0})) \oplus 0 \cdots 01$. Let $\beta_{2}=F_{0}(k_{0} \oplus 0 \cdots 01) \oplus F_{0}(k_{0})$, then $b_{3}=F_{1}(\beta_{2}) \oplus 0 \cdots 01$, the value of $\beta_{2}$ can be obtained using Grover algorithm in $O(2^{n/2})$ time.

After getting the value of $\beta_{2}=F_{0}(k_{0} \oplus 0 \cdots 01) \oplus F_{0}(k_{0})$, $k_{0}$ can be obtained by Grover algorithm in $O(2^{n/2})$ time, and then $k_{1}=\beta_{1} \oplus F_{0}(k_{0})$.

Now, querying $(a_{0}, b_{0})=(F_{1}(k_{1}), F_{0}(k_{0} \oplus F_{1}(k_{1})))$ to obtain $a_{3}=F_{2}(k_{2})$, and $k_{2}$ can be obtained by Grover algorithm in $O(2^{n/2})$ time.

Finally, by choosing a constant number of plaintexts and querying the encryption oracle to obtain the corresponding ciphertexts, one can verify whether the plaintext-ciphertext pairs produced by the recovered keys \((k_0, k_1, k_2)\) are consistent with the 3-round Feistel-KF structure.

\section{Quantum attacks on FBC-F/KF/FK structures in the Q2 model}\label{sec3}

In this section, we present quantum key-recovery attacks on FBC-F, FBC-KF, and FBC-FK structures in the Q2 model. First, we construct 4-round quantum distinguishers for FBC-F and FBC-KF structures, and a 6-round quantum distinguisher for FBC-FK structure. Building on these results, we apply the Grover-meets-Simon algorithm to recover partial subkeys of the $r$-round encryption structures.

\par \subsection{Quantum attacks on FBC-F/KF structures}

Under the qCPA setting, we first construct a 4-round quantum distinguisher and give an $r$-round quantum key-recovery attack on FBC-F structure. The quantum attack on the FBC-KF structure is similar to that on the FBC-F structure. Both attacks have the same time complexity.

\textit{A1) Quantum distinguishers against the 4-round FBC-F/KF structures}

For the 4-round FBC-F structure, we design a new periodic function that allows Simon's algorithm to distinguish it from a random permutation in $O(n)$ time, where $n$ denotes the key length.

\begin{figure}
    \centering
    \includegraphics[width=0.7\linewidth]{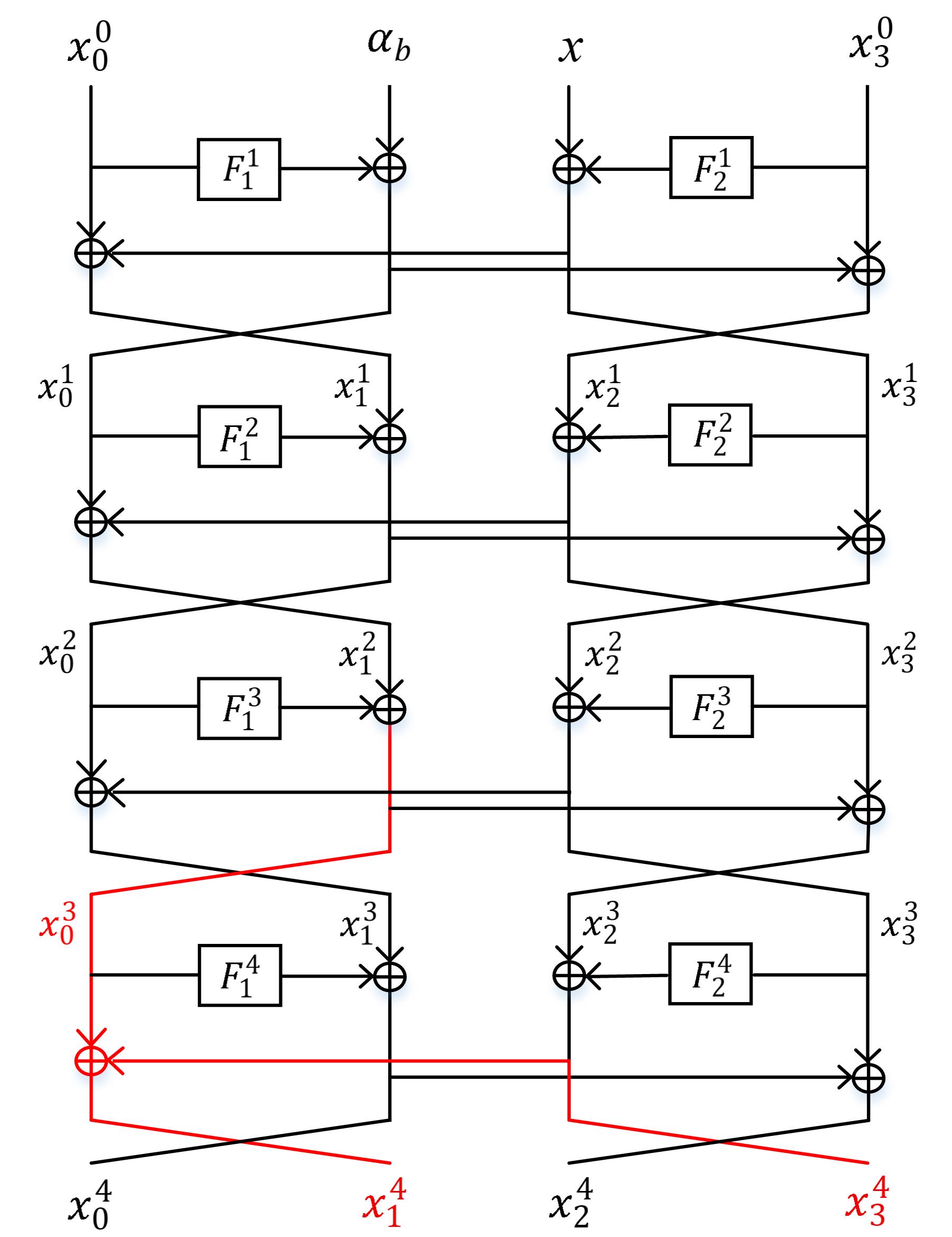}
    \caption{4-round FBC-F structure.}
    \label{fig7}
\end{figure}

Let $\alpha_{0}, \alpha_{1}, x_{0}^{0}, x_{3}^{0} \in \{0,1\}^{n}$ be arbitrary constants, where $\alpha_{0}\neq \alpha_{1}$. Given a 4-round FBC-F encryption oracle $\mathcal{O}_{1}$, we can define
\[
    \begin{aligned}
        f^{\mathcal{O}_{1}}:\{0,1\}^{n} &\rightarrow\{0,1\}^{n} \\
x &\mapsto x_{1}^{4}(\alpha_{0}) \oplus x_{3}^{4}(\alpha_{0}) \oplus x_{1}^{4}(\alpha_{1}) \oplus x_{3}^{4}(\alpha_{1}).
    \end{aligned}
\]
Here, \( x_i^j(\alpha_b) \) denotes the output of the \((i+1)\)-th branch in the \( j \)-th round of the FBC-F structure when the input plaintext is \( (x_0^0, \alpha_b, x, x_3^0) \), where \( i \in \{0,1,2,3\} \), \( j \in \{1,2,3,4\} \), and \( b \in \{0,1\} \).

From the 4-round FBC-F structure (see Fig. \ref{fig7}), we obtain
\begin{equation}
    \begin{aligned}
f^{\mathcal{O}_{1}}(x)
& =x_{1}^{4}(\alpha_{0}) \oplus x_{3}^{4}(\alpha_{0}) \oplus x_{1}^{4}(\alpha_{1}) \oplus x_{3}^{4}(\alpha_{1})\\
& =x_{0}^{3}(\alpha_{0}) \oplus x_{0}^{3}(\alpha_{1}) .
\end{aligned}
\end{equation}
Next, we have the following lemma.

\begin{Lemma}
If $\mathcal{O}_1$ is the encryption algorithm of $4$-round FBC-F structure, for any $x \in \{0,1\}^n$, the function $f^{\mathcal{O}_1}$ satisfies
\[f^{\mathcal{O}_1}(x)=f^{\mathcal{O}_1}(x \oplus F_{1}^{2}(F_{1}^{1}(x_{0}^{0}) \oplus \alpha_{0}) \oplus F_{1}^{2}(F_{1}^{1}(x_{0}^{0}) \oplus \alpha_{1})). \]
That is, $f^{\mathcal{O}_1}$ has the period $s=F_{1}^{2}(F_{1}^{1}(x_{0}^{0}) \oplus \alpha_{0}) \oplus F_{1}^{2}(F_{1}^{1}(x_{0}^{0}) \oplus \alpha_{1})$.
\end{Lemma}
\begin{proof}
By choosing the plaintext $(x_{0}^{0}, \alpha_{b}, x, x_{3}^{0})$ as the input of the FBC-F structure, the output of the first round is
\begin{equation}\label{Eq2}
\begin{aligned}
(x_{0}^{1}, x_{1}^{1}, x_{2}^{1}, x_{3}^{1})=&(F_{1}^{1}(x_{0}^{0}) \oplus \alpha_{b}, x_{0}^{0} \oplus x \oplus F_{2}^{1}(x_{3}^{0}), x_{3}^{0} \oplus\\ &F_{1}^{1}(x_{0}^{0}) \oplus \alpha_{b}, x \oplus F_{2}^{1}(x_{3}^{0})).
\end{aligned}
\end{equation}
Similarly, the output of the second round $(x_{0}^{2}, x_{1}^{2}, x_{2}^{2}, x_{3}^{2})$ can be obtained, i.e.,
\begin{equation}\label{Eq3}
  \begin{aligned}
x_{0}^{2}&=F_{1}^{2}(F_{1}^{1}(x_{0}^{0}) \oplus \alpha_{b}) \oplus x \oplus x_{0}^{0} \oplus F_{2}^{1}(x_{3}^{0}), \\
x_{1}^{2}&=x_{3}^{0} \oplus F_{2}^{2}(x \oplus F_{2}^{1}(x_{3}^{0})), \\
x_{2}^{2}&=x_{0}^{0} \oplus F_{1}^{2}(\alpha_{b} \oplus F_{1}^{1}( x_{0}^{0})), \\
x_{3}^{2}&=x_{3}^{0} \oplus F_{1}^{1}(x_{0}^{0}) \oplus \alpha_{b} \oplus F_{2}^{2}(x \oplus F_{2}^{1}(x_{3}^{0})).
\end{aligned}
\end{equation}
As illustrated by the red line in Fig. \ref{fig7}, using the ciphertexts \( x_{1}^4 \) and \( x_{3}^4 \) obtained from the encryption oracle $\mathcal{O}_1$ of the 4-round FBC-F structure, we can derive the value of \( x_{0}^3 \) based on the relation \( x_{1}^4  \oplus x_{3}^4 = x_{0}^3 \).
On the other hand, according to the encryption structure of FBC-F, the intermediate value $x_0^3$ can be expressed as
\begin{equation}\label{Eq4}
\begin{aligned}
x_{0}^{3}=
F_{1}^{3}(x_{0}^{2}) \oplus x_{1}^{2}.
\end{aligned}
\end{equation}
Therefore, it can be seen from Eqs. (\ref{Eq3}) and (\ref{Eq4}) that
\[
\begin{aligned}
x_{0}^{3}=&x_{3}^{0} \oplus F_{2}^{2}(x \oplus F_{2}^{1}(x_{3}^{0})) \oplus F_{1}^{3}(F_{1}^{2}(F_{1}^{1}(x_{0}^{0}) \oplus \alpha_{b})\\ & \oplus x \oplus x_{0}^{0} \oplus F_{2}^{1}(x_{3}^{0})).
\end{aligned}
\]
Hence, the function $f^{\mathcal{O}_1}$ can be described as
\[
    \begin{aligned}
f^{\mathcal{O}_1}(x)&=x_{0}^{3}(\alpha_{0}) \oplus x_{0}^{3}(\alpha_{1}) \\
&=F_{1}^{3}(F_{1}^{2}(F_{1}^{1}(x_{0}^{0}) \oplus \alpha_{0}) \oplus x \oplus x_{0}^{0} \oplus F_{2}^{1}(x_{3}^{0}))\\&\quad \oplus F_{1}^{3}(F_{1}^{2}(F_{1}^{1}(x_{0}^{0}) \oplus \alpha_{1}) \oplus x \oplus x_{0}^{0} \oplus F_{2}^{1}(x_{3}^{0})).
\end{aligned}
\]
The function $f^{\mathcal{O}_1}$ has a period $s$ since it satisfies
\[
\begin{aligned}
f&^{\mathcal{O}_1}(x \oplus F_{1}^{2}(F_{1}^{1}(x_{0}^{0}) \oplus \alpha_{0}) \oplus F_{1}^{2}(F_{1}^{1}(x_{0}^{0}) \oplus \alpha_{1}) )\\
&=F_{1}^{3}(\textcolor{red}{F_{1}^{2}(F_{1}^{1}(x_{0}^{0}) \oplus \alpha_{0})} \oplus x \oplus \textcolor{red}{F_{1}^{2}(F_{1}^{1}(x_{0}^{0}) \oplus \alpha_{0})} \oplus\\
&\quad F_{1}^{2}(F_{1}^{1}(x_{0}^{0}) \oplus \alpha_{1}) \oplus x_{0}^{0} \oplus F_{2}^{1}(x_{3}^{0})) \oplus F_{1}^{3}(\textcolor{blue}{F_{1}^{2}(F_{1}^{1}(x_{0}^{0})}\\
&\quad \textcolor{blue}{\oplus \alpha_{1})} \oplus x \oplus F_{1}^{2}(F_{1}^{1}(x_{0}^{0}) \oplus \alpha_{0}) \oplus \textcolor{blue}{F_{1}^{2}(F_{1}^{1}(x_{0}^{0}) \oplus \alpha_{1})} \oplus \\
&\quad x_{0}^{0} \oplus F_{2}^{1}(x_{3}^{0})) \\
&=F_{1}^{3}(x \oplus F_{1}^{2}(F_{1}^{1}(x_{0}^{0}) \oplus \alpha_{1}) \oplus x_{0}^{0} \oplus F_{2}^{1}(x_{3}^{0})) \oplus F_{1}^{3}(x\\
&\quad \oplus F_{1}^{2}(F_{1}^{1}(x_{0}^{0}) \oplus \alpha_{0}) \oplus x_{0}^{0} \oplus F_{2}^{1}(x_{3}^{0})) \\
&=f^{\mathcal{O}_1}(x),
\end{aligned}
\]
where $s=F_{1}^{2}(F_{1}^{1}(x_{0}^{0}) \oplus \alpha_{0}) \oplus F_{1}^{2}(F_{1}^{1}(x_{0}^{0}) \oplus \alpha_{1}).$ The proof of the above lemma holds.
\end{proof}

Therefore, we can construct a distinguisher for the 4-round FBC-F structure by using the function $f^{\mathcal{O}_1}$. Moreover, the Simon algorithm introduced in Sect. \ref{sec2} can efficiently find this period $s$ in $O(n)$ time.

In the FBC-KF structure, a 4-round quantum distinguisher can also be constructed using a similar approach. The detailed procedure is provided in Appendix \ref{A}.

\textit{A2) Quantum key-recovery attacks on FBC-F/KF structures}

Here, we design a quantum key-recovery attack on FBC-F structure using Grover-meets-Simon algorithm. We extend the $4$-round FBC-F structure to 6 rounds, as shown in Fig. \ref{fig8}. To apply the 4-round quantum distinguisher for the FBC-F structure constructed in the previous section, we need to obtain the values of \(x_1^4\) and \(x_3^4\).

First, we query the encryption oracle $\mathcal{O}_2$ of the 6-round FBC-F structure with the plaintext $(x_{0}^{0}, \alpha_{b}, x, x_{3}^{0})$ to obtain the ciphertext $(x_{0}^{6}, x_{1}^{6}, x_{2}^{6}, x_{3}^{6})$. In Fig. \ref{fig8}, the decryption paths of $x_{1}^{4}$ and $x_{3}^{4}$ are highlighted, where the red line corresponds to the path of $x_{1}^{4}$ and the blue line corresponds to that of $x_{3}^{4}$. We can calculate

\begin{equation}\label{Eq5}
\begin{cases}
\begin{aligned}
x_{1}^{4}&=x_{0}^{5} \oplus F_{1}^{5}(x_{1}^{5} \oplus x_{3}^{5})\\
x_{3}^{4}&=x_{0}^{5} \oplus x_{2}^{5}
\end{aligned}
\end{cases}
\end{equation}
and
\begin{equation}\label{Eq6}
\begin{cases}
\begin{aligned}
x_{0}^{5}&=x_{1}^{6} \oplus x_{3}^{6} \\
x_{1}^{5}&=x_{0}^{6} \oplus F_{1}^{6}(x_{1}^{6} \oplus x_{3}^{6})\\
x_{2}^{5}&=x_{3}^{6} \oplus F_{2}^{6}(x_{0}^{6} \oplus x_{2}^{6})\\
x_{3}^{5}&=x_{0}^{6} \oplus x_{2}^{6}
\end{aligned}.
\end{cases}
\end{equation}

\begin{figure}[t]
    \centering
    \includegraphics[width=0.7\linewidth]{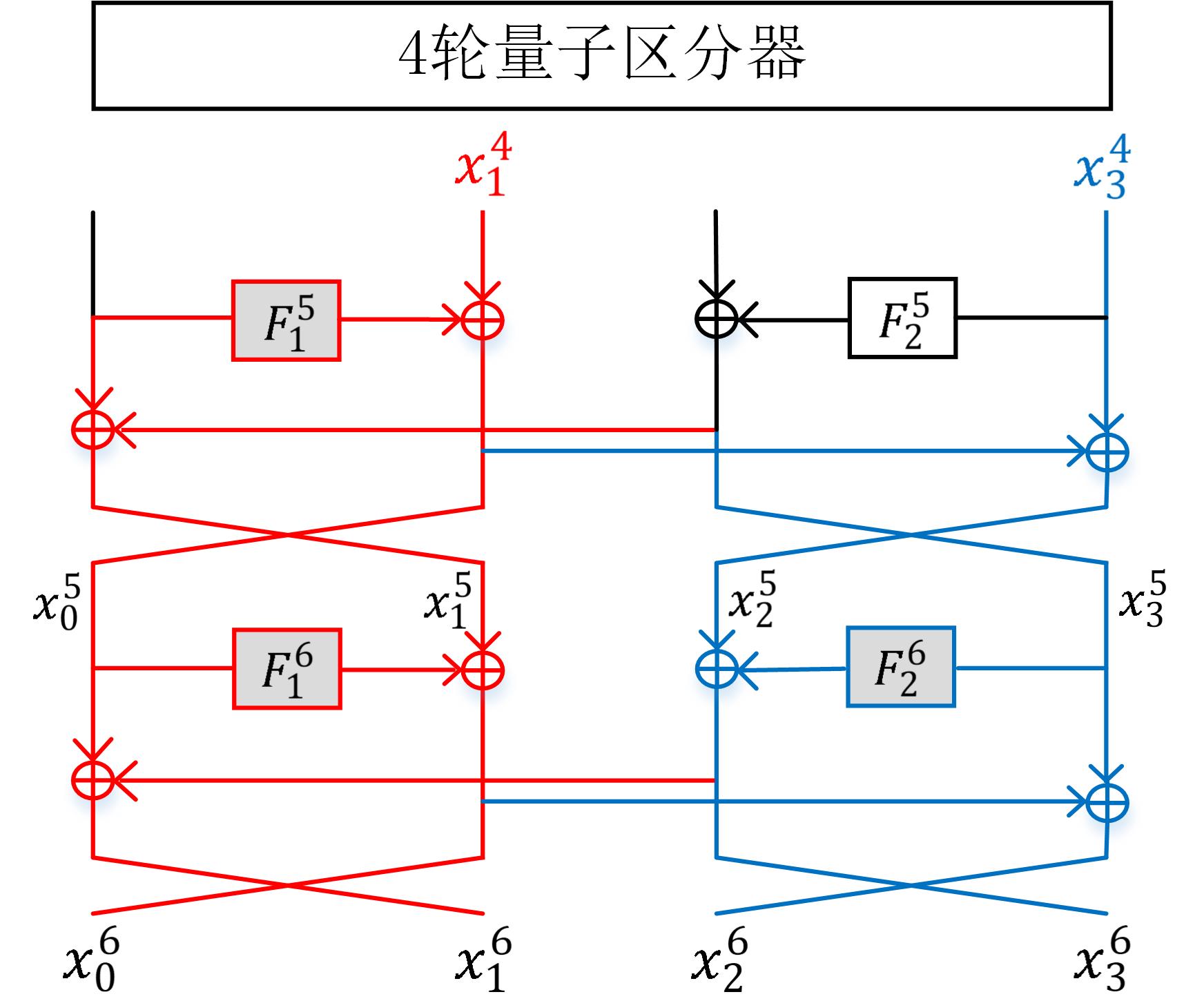}
    \caption{6-round quantum key-recovery attack on FBC-F structure.}
    \label{fig8}
\end{figure}

Second, according to Eqs. (\ref{Eq5}) and (\ref{Eq6}), we need to guess the subkeys $k_{1}^{5 }, k_{1}^{6 } \text { and } k_{2}^{6}$ to calculate $x_{1}^{4}$ and $x_{3}^{4}$. For guessed subkeys, we apply the 4-round distinguisher to check its correctness. Specifically, the adversary can obtain a periodic function only if the guess is correct, which can then be detected using the Simon algorithm.

Attacking the 6-round FBC-F structure based on the 4-round quantum distinguisher requires guessing $3n$ bits of subkeys using the Grover algorithm, and the time complexity is $O(2^{3n/2})$. For $r>6$ rounds, using the Grover algorithm requires guessing $2n(r-6)+3n$ bits of subkeys, and the time complexity is $O(2^{(2n(r-6)+3n)/2})$, which is a factor of $2^{4.5n}$ lower than that of quantum brute-force search.

Similarly, the quantum key-recovery attack on FBC-KF structure is the same as that on FBC-F structure, with the same time complexity.

\par \subsection{Quantum attack on FBC-FK structure}
Here, we construct a 6-round quantum distinguisher for the FBC-FK structure and perform an $r(r>6)$-round quantum key-recovery attack.

\textit{B1) Quantum distinguisher against the 6-round FBC-FK structure}

To construct the 6-round quantum distinguisher, we choose the plaintext input to the FBC-FK structure as $(\alpha _{b} \oplus x_3^0,F_1^1 (\alpha _{b}\oplus x_3^0)\oplus x_0^0,x_3^0\oplus F_2^1 (x_0^0\oplus x),x_0^0\oplus x)$, where  $b\in \{0,1\}$, $x\in \{0,1\}^n$, and $\alpha_{0}$, $\alpha_{1}$, $x_{0}^{0}$, $x_{3}^{0}$ are constants such that $\alpha_{0}\neq  \alpha_{1}$.

\begin{figure}
    \centering
    \includegraphics[width=0.7\linewidth]{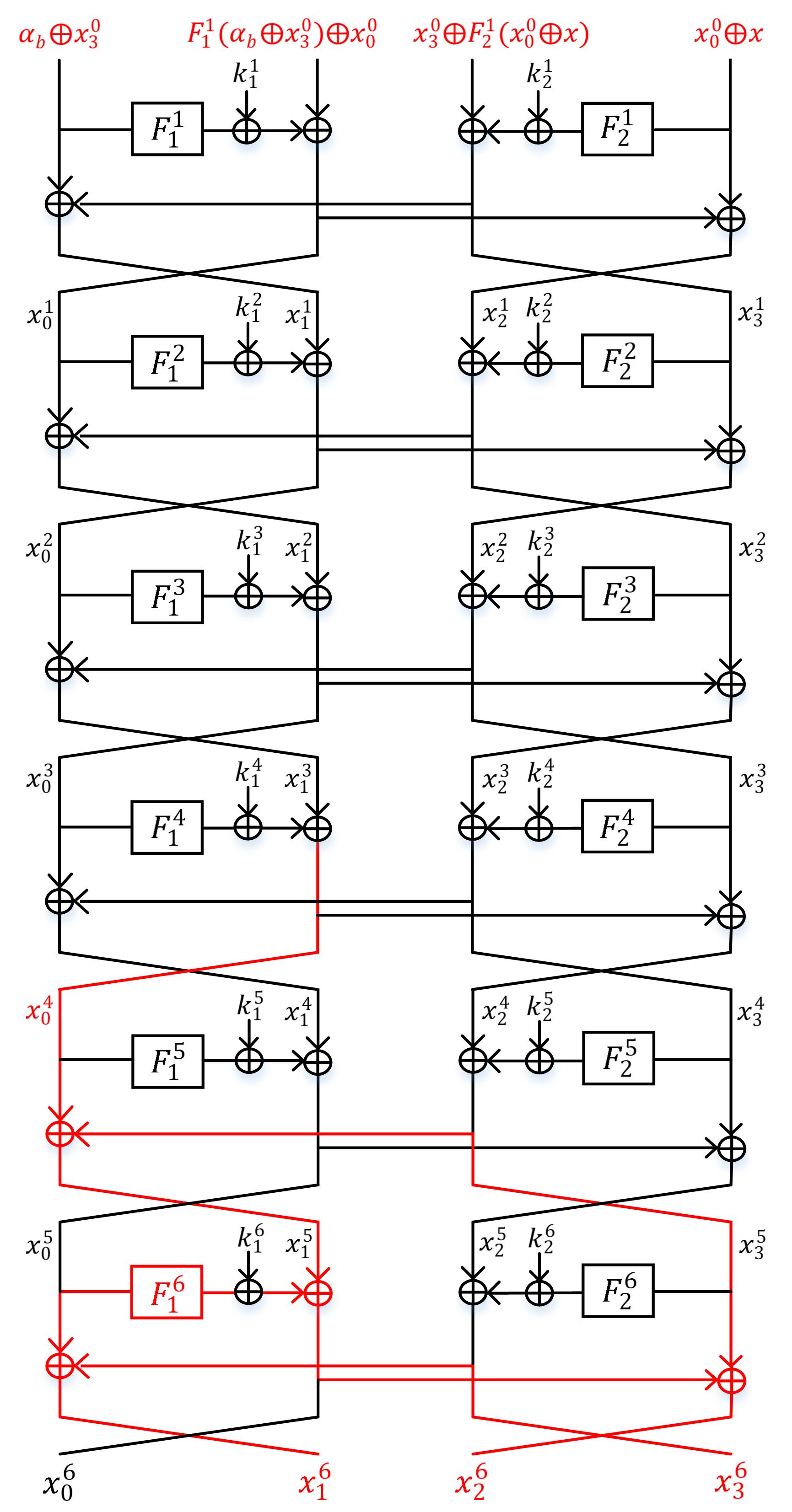}
    \caption{6-round FBC-FK structure.}
    \label{fig9}
\end{figure}

Based on this plaintext, we obtain the ciphertext $(x_0^6,x_1^6,x_2^6,x_3^6)$ through the 6-round FBC-FK encryption oracle $\mathcal{O}_3$ as shown in Fig. \ref{fig9}. Here, we define
\[
\begin{aligned}
f^{\mathcal{O}_3}:\{0,1\}^{n} & \rightarrow\{0,1\}^{n} \\
x & \mapsto\left[F_{1}^{6}\left(x_{1}^{6} \oplus x_{3}^{6}\right) \oplus x_{2}^{6}\right]\left(\alpha_{0}\right) \\
&\quad\quad \oplus \left[F_{1}^{6}\left(x_{1}^{6} \oplus x_{3}^{6}\right) \oplus x_{2}^{6}\right]\left(\alpha_{1}\right),
\end{aligned}
\]
where the expression
$ \left[F_{1}^{6}(x_{1}^{6} \oplus x_{3}^{6}) \oplus x_{2}^{6} \right](\alpha_b) $
denotes the value obtained by certain outputs from the $6$-th round of the FBC-FK structure, given the plaintext input $(\alpha _{b} \oplus x_3^0,F_1^1 (\alpha _{b}\oplus x_3^0)\oplus x_0^0,x_3^0\oplus F_2^1 (x_0^0\oplus x),x_0^0\oplus x)$ for $b\in \{0,1\}$. The functions \( F_i^j \) are publicly known for all \( i \in \{1, 2\} \) and \( j \in \{1, 2, 3, 4, 5, 6\} \).

Now, we give the following \autoref{Lemma2}.
\begin{Lemma} \label{Lemma2}
If $\mathcal{O}_3$ is the encryption algorithm of $6$-round FBC-FK structure, then for any $x \in \{0,1\}^n$, the function $f^{\mathcal{O}_3}$ satisfies
\begin{align*}
f^{\mathcal{O}_3}(x)&=f^{\mathcal{O}_3}(x \oplus F_{1}^{3}(\alpha_{0} \oplus k_{2}^{1} \oplus k_{1}^{2} \oplus F_{1}^{2}(x_{0}^{0} \oplus k_{1}^{1}))\\
 &\quad \oplus F_{1}^{3}(\alpha_{1} \oplus k_{2}^{1} \oplus k_{1}^{2} \oplus F_{1}^{2}(x_{0}^{0} \oplus k_{1}^{1}))).
 \end{align*}
That is, $f^{\mathcal{O}_3}$ has the period $s=F_{1}^{3}(\alpha_{0} \oplus k_{2}^{1} \oplus k_{1}^{2} \oplus F_{1}^{2}(x_{0}^{0} \oplus k_{1}^{1})) \oplus F_{1}^{3}(\alpha_{1} \oplus k_{2}^{1} \oplus k_{1}^{2} \oplus F_{1}^{2}(x_{0}^{0} \oplus k_{1}^{1})).$
\end{Lemma}

\begin{proof}
 Choosing the plaintext $(\alpha_{b} \oplus x_{3}^{0}, F_{1}^{1}(\alpha_{b} \oplus x_{3}^{0}) \oplus x_{0}^{0}, x_{3}^{0} \oplus F_{2}^{1}(x_{0}^{0} \oplus x), x_{0}^{0} \oplus x)$ as the input of the 6-round FBC-FK structure (see Fig. \ref{fig9}), we find that the output of the first round is
\begin{equation}
\begin{aligned}
(x_{0}^{1}, x_{1}^{1}, x_{2}^{1}, x_{3}^{1})
& =(x_{0}^{0} \oplus k_{1}^{1}, \alpha_{b} \oplus k_{2}^{1}, x \oplus k_{1}^{1}, x_{3}^{0} \oplus k_{2}^{1}).
\end{aligned}
\end{equation}
Next, we obtain the output of the second round as $(x_0^2,x_1^2,x_2^2,x_3^2)$, where
\begin{equation}
  \begin{aligned}
x_{0}^{2}&=\alpha_{b} \oplus k_{2}^{1} \oplus k_{1}^{2} \oplus F_{1}^{2}(x_{0}^{0} \oplus k_{1}^{1}), \\
x_{1}^{2}&=x \oplus x_{0}^{0} \oplus k_{2}^{2} \oplus F_{2}^{2}(x_{3}^{0} \oplus k_{2}^{1}), \\
x_{2}^{2}&=\alpha_{b} \oplus x_{3}^{0} \oplus k_{1}^{2} \oplus F_{1}^{2}(x_{0}^{0} \oplus k_{1}^{1}), \\
x_{3}^{2}&=x \oplus k_{1}^{1} \oplus k_{2}^{2} \oplus F_{2}^{2}(x_{3}^{0} \oplus k_{2}^{1}).
\end{aligned}
\end{equation}
Similarly, the third round output $(x_0^3,x_1^3,x_2^3,x_3^3)$ can be expressed as
\begin{equation}
    \begin{aligned}
x_{0}^{3}&=x \oplus x_{0}^{0} \oplus k_{2}^{2} \oplus k_{1}^{3} \oplus F_{2}^{2}(x_{3}^{0} \oplus k_{2}^{1})\\
 &\quad \oplus F_{1}^{3}(\alpha_{b} \oplus k_{2}^{1} \oplus k_{1}^{2} \oplus F_{1}^{2}(x_{0}^{0} \oplus k_{1}^{1})), \\
x_{1}^{3}&=x_{3}^{0} \oplus k_{2}^{1} \oplus k_{2}^{3} \oplus F_{2}^{3}(x \oplus k_{1}^{1} \oplus k_{2}^{2} \oplus F_{2}^{2}(x_{3}^{0} \oplus k_{2}^{1})), \\
x_{2}^{3}&=x_{0}^{0} \oplus k_{1}^{1} \oplus k_{1}^{3} \oplus F_{1}^{3}(\alpha_{b} \oplus k_{2}^{1} \oplus k_{1}^{2} \oplus F_{1}^{2}(x_{0}^{0} \oplus k_{1}^{1})), \\
x_{3}^{3}&=\alpha_{b} \oplus x_{3}^{0} \oplus k_{1}^{2} \oplus k_{2}^{3} \oplus F_{1}^{2}(x_{0}^{0} \oplus k_{1}^{1})\\
 &\quad\oplus F_{2}^{3}(x \oplus k_{1}^{1} \oplus k_{2}^{2} \oplus F_{2}^{2}(x_{3}^{0} \oplus k_{2}^{1})) .
\end{aligned}
\end{equation}
Consider the intermediate state $x_0^4$, where $x_0^4=x_1^3\oplus k_1^4\oplus F_1^4 (x_0^3)$. We can then get
\begin{equation}\label{Eq11}
    \begin{aligned}
x_{0}^{4}= & x_{3}^{0} \oplus k_{2}^{1} \oplus k_{2}^{3} \oplus k_{1}^{4}
\oplus F_{2}^{3}(x \oplus k_{1}^{1} \oplus k_{2}^{2} \oplus F_{2}^{2}(x_{3}^{0} \\ &\oplus k_{2}^{1}))
\oplus F_{1}^{4}(x \oplus x_{0}^{0} \oplus k_{2}^{2} \oplus k_{1}^{3}
\oplus F_{2}^{2}(x_{3}^{0} \oplus k_{2}^{1}) \\ &\oplus F_{1}^{3}(\alpha_{b} \oplus k_{2}^{1} \oplus k_{1}^{2}
\oplus F_{1}^{2}(x_{0}^{0} \oplus k_{1}^{1}))).
\end{aligned}
\end{equation}
From Fig. \ref{fig9} (red line), $x_0^4$ can also be described as
\[
    \begin{aligned}
x_{0}^{4}=x_{1}^{5}\oplus x_{3}^{5}=F_{1}^{6}(x_{1}^{6}\oplus x_{3}^{6})\oplus k_{1}^{6} \oplus x_{2}^{6},
\end{aligned}
\]
so we can get
\begin{equation}\label{Eq12}
    \begin{aligned}
x_{0}^{4}\oplus k_{1}^{6}=F_{1}^{6}(x_{1}^{6}\oplus x_{3}^{6})\oplus x_{2}^{6}.
\end{aligned}
\end{equation}
In particular, from Eqs.  (\ref{Eq11}) and (\ref{Eq12}), we have
\begin{equation}
    \begin{aligned}
f^{\mathcal{O}_3}&=
[F_{1}^{6}(x_{1}^{6} \oplus x_{3}^{6}) \oplus x_{2}^{6}](\alpha_{0})
\oplus [F_{1}^{6}(x_{1}^{6} \oplus x_{3}^{6})\oplus x_{2}^{6}](\alpha_{1}) \\
&= x_{0}^{4}(\alpha_{0}) \oplus x_{0}^{4}(\alpha_{1}) \\
&= F_{1}^{4}(x \oplus x_{0}^{0} \oplus k_{2}^{2} \oplus k_{1}^{3}
\oplus F_{2}^{2}(x_{3}^{0} \oplus k_{2}^{1})\oplus F_{1}^{3}(\alpha_{0} \oplus  k_{2}^{1}\\
&\hspace{0.4cm}
 \oplus k_{1}^{2}
\oplus F_{1}^{2}(x_{0}^{0} \oplus k_{1}^{1})))\oplus
F_{1}^{4}(x \oplus x_{0}^{0} \oplus k_{2}^{2}\oplus k_{1}^{3}\oplus \\
&\hspace{0.4cm}
F_{2}^{2}(x_{3}^{0} \oplus k_{2}^{1})\oplus F_{1}^{3}(\alpha_{1} \oplus k_{2}^{1} \oplus k_{1}^{2}
\oplus F_{1}^{2}(x_{0}^{0} \oplus k_{1}^{1}))).
\end{aligned}
\end{equation}
The function $f^{\mathcal{O}_3}$ has a period because it satisfies
\[
 \begin{aligned}
f^{\mathcal{O}_3}&(x \oplus F_{1}^{3}(\alpha_{0} \oplus k_{2}^{1} \oplus k_{1}^{2} \oplus F_{1}^{2}(x_{0}^{0} \oplus k_{1}^{1}))\\
&\hspace{2cm}\oplus F_{1}^{3}(\alpha_{1} \oplus k_{2}^{1} \oplus k_{1}^{2} \oplus F_{1}^{2}(x_{0}^{0} \oplus k_{1}^{1})))\\
&= F_{1}^{4}(x \oplus \textcolor{red}{F_{1}^{3}(\alpha_{0} \oplus k_{2}^{1} \oplus k_{1}^{2} \oplus F_{1}^{2}(x_{0}^{0} \oplus k_{1}^{1}))}\oplus F_{1}^{3}(\alpha_{1}\\
&\hspace{0.4cm} \oplus k_{2}^{1} \oplus k_{1}^{2} \oplus F_{1}^{2}(x_{0}^{0} \oplus k_{1}^{1})) \oplus x_{0}^{0} \oplus k_{2}^{2}\oplus k_{1}^{3} \oplus F_{2}^{2}(x_{3}^{0}\\
&\hspace{0.4cm}\oplus k_{2}^{1})
\oplus \textcolor{red}{F_{1}^{3}(\alpha_{0} \oplus k_{2}^{1} \oplus k_{1}^{2} \oplus F_{1}^{2}(x_{0}^{0} \oplus k_{1}^{1}))})\oplus F_{1}^{4}(x\\
&\hspace{0.4cm} \oplus F_{1}^{3}(\alpha_{0} \oplus k_{2}^{1} \oplus k_{1}^{2}  \oplus F_{1}^{2}(x_{0}^{0} \oplus k_{1}^{1}))
\oplus \textcolor{blue}{F_{1}^{3}(\alpha_{1}\oplus k_{2}^{1}}\\
&\hspace{0.5cm}\textcolor{blue}{ \oplus k_{1}^{2} \oplus  F_{1}^{2}(x_{0}^{0} \oplus k_{1}^{1})) }
\oplus x_{0}^{0}  \oplus k_{2}^{2} \oplus k_{1}^{3} \oplus F_{2}^{2}(x_{3}^{0} \oplus k_{2}^{1})\\
&\hspace{0.4cm}
\oplus \textcolor{blue}{F_{1}^{3}(\alpha_{1} \oplus k_{2}^{1} \oplus k_{1}^{2} \oplus F_{1}^{2}(x_{0}^{0}\oplus k_{1}^{1}))}) \\
&= F_{1}^{4}(x \oplus F_{1}^{3}(\alpha_{1} \oplus k_{2}^{1} \oplus k_{1}^{2} \oplus F_{1}^{2}(x_{0}^{0} \oplus k_{1}^{1}))
\oplus x_{0}^{0} \oplus k_{2}^{2}\\
&\hspace{0.4cm} \oplus k_{1}^{3} \oplus F_{2}^{2}(x_{3}^{0} \oplus k_{2}^{1})) \oplus  F_{1}^{4}(x \oplus F_{1}^{3}(\alpha_{0} \oplus k_{2}^{1}\oplus k_{1}^{2}\\
&\hspace{0.4cm}  \oplus F_{1}^{2}(x_{0}^{0} \oplus k_{1}^{1}))
\oplus x_{0}^{0} \oplus k_{2}^{2} \oplus k_{1}^{3} \oplus F_{2}^{2}(x_{3}^{0} \oplus k_{2}^{1})) \\
&= f^{\mathcal{O}_3}(x),
 \end{aligned}
\]
where $s=F_{1}^{3}(\alpha_{0} \oplus k_{2}^{1} \oplus k_{1}^{2} \oplus F_{1}^{2}(x_{0}^{0} \oplus k_{1}^{1}))
\oplus F_{1}^{3}(\alpha_{1} \oplus k_{2}^{1} \oplus k_{1}^{2} \oplus F_{1}^{2}(x_{0}^{0} \oplus k_{1}^{1}))$. Thus the \autoref{Lemma2} follows.
\end{proof}
Therefore, we can construct a quantum distinguisher for the 6-round FBC-FK structure and recover the period $s$ in $O(n)$ time using Simon algorithm.

\textit{B2) Quantum key-recovery attack on FBC-FK structure}

Using the 6-round quantum distinguisher, we can extend a key-recovery attack to the $r(r>6)$-round FBC-FK structure with the time complexity $O(2^{n(r-6)})$.

Specifically, we construct a key-recovery attack on the $r$-round FBC-FK by appending $(r-6)$ rounds to the 6-round distinguisher.

(1). Query the $r$-round FBC-FK encryption oracle $\mathcal{O}_4$ with the plaintext $(\alpha_{b} \oplus x_{3}^{0}, F_{1}^{1}(\alpha_{b} \oplus x_{3}^{0}) \oplus x_{0}^{0}, x_{3}^{0} \oplus F_{2}^{1}(x_{0}^{0} \oplus x), x_{0}^{0} \oplus x)$ to receive the output $(x_{0}^{r}, x_{1}^{r}, x_{2}^{r}, x_{3}^{r})$. Use the subkeys from the last $(r-6)$ rounds as input, and decrypt these rounds to obtain the intermediate state values $(x_{0}^{6}, x_{1}^{6}, x_{2}^{6}, x_{3}^{6})$.

(2). Guess the subkeys for the last $(r-6)$ rounds. For the guessed subkeys, apply the 6-round distinguisher to verify their correctness. Only correct guesses yield the periodic function described in \autoref{Lemma2}, which can then be detected using Simon algorithm.

For the $r(r>6)$-round FBC-FK structure, we need to guess the $2n(r-6)$ bits using Grover algorithm. Therefore, the subkeys of the $r$-round FBC-FK structure can be recovered with $O(2^{n(r-6)})$ quantum queries. It reduces the time complexity by a factor of $2^{6n}$ compared with quantum brute-force search.

\section{Low-data key-recovery attacks on FBC-KF/FK structures in the Q1 model}\label{sec4}

Inspired by the work of Daiza et al. \cite{Daiza2022}, we propose low-data key-recovery attacks on FBC-KF and FBC-FK structures in the Q1 model. Our attacks operate in the qCPA setting and require only $O(1)$ plaintext-ciphertext pairs. By analyzing the encryption process of the FBC-KF/FK structures and subsequently applying Grover algorithm to search for specific intermediate states, all keys can be efficiently recovered in $O(2^{n/2})$ time.

\subsection{Key-recovery attack on 4-round FBC-KF structure}\label{4.1}

Assume that the adversary has only classical query access and quantum computing capabilities. We perform a quantum key-recovery attack on 4-round FBC-KF structure under the qCPA setting.

According to the 4-round FBC-KF structure shown in Fig. \ref{fig10} (red line), we can easily see that
\begin{equation}\label{Eq14}
\begin{aligned}
x_0^2=x_1^3\oplus x_3^3=x_2^4\oplus F_1^4 (k_1^4\oplus x_1^4\oplus x_3^4).
\end{aligned}
\end{equation}
Similarly, from the decryption perspective, we can describe the output value $x_{0}^{3}$ as $x_{0}^{3}=x_1^4\oplus x_3^4$. However, when considering the encryption direction, $x_{0}^{3}$ can also be expressed as
\begin{equation}\label{Eq15}
\begin{aligned}
x_{0}^{3}&=x_3^0\oplus F_2^2 (x_2^0\oplus k_2^2\oplus F_2^1 (x_3^0\oplus k_2^1))\oplus F_1^3 (x_0^0\oplus x_2^0 \oplus \\
&\quad k_1^3\oplus F_2^1 (x_3^0\oplus k_2^1)\oplus
F_1^2 (x_1^0\oplus k_1^2\oplus F_1^1 (x_0^0\oplus k_1^1))).
\end{aligned}
\end{equation}

By analyzing the intermediate process of the FBC-KF structure (see Fig. \ref{fig10}), we present the following lemma. Note that the adversary's quantum capabilities are restricted to the Q1 model.

\begin{Lemma}
Given a 4-round FBC-KF encryption oracle, where $F$ is a public function and the subkey length is $n$ bits, the adversary can recover all keys in $O(2^{n/2})$ time using Grover algorithm in the Q1 model.
\end{Lemma}

\begin{figure}
    \centering
    \includegraphics[width=0.7\linewidth]{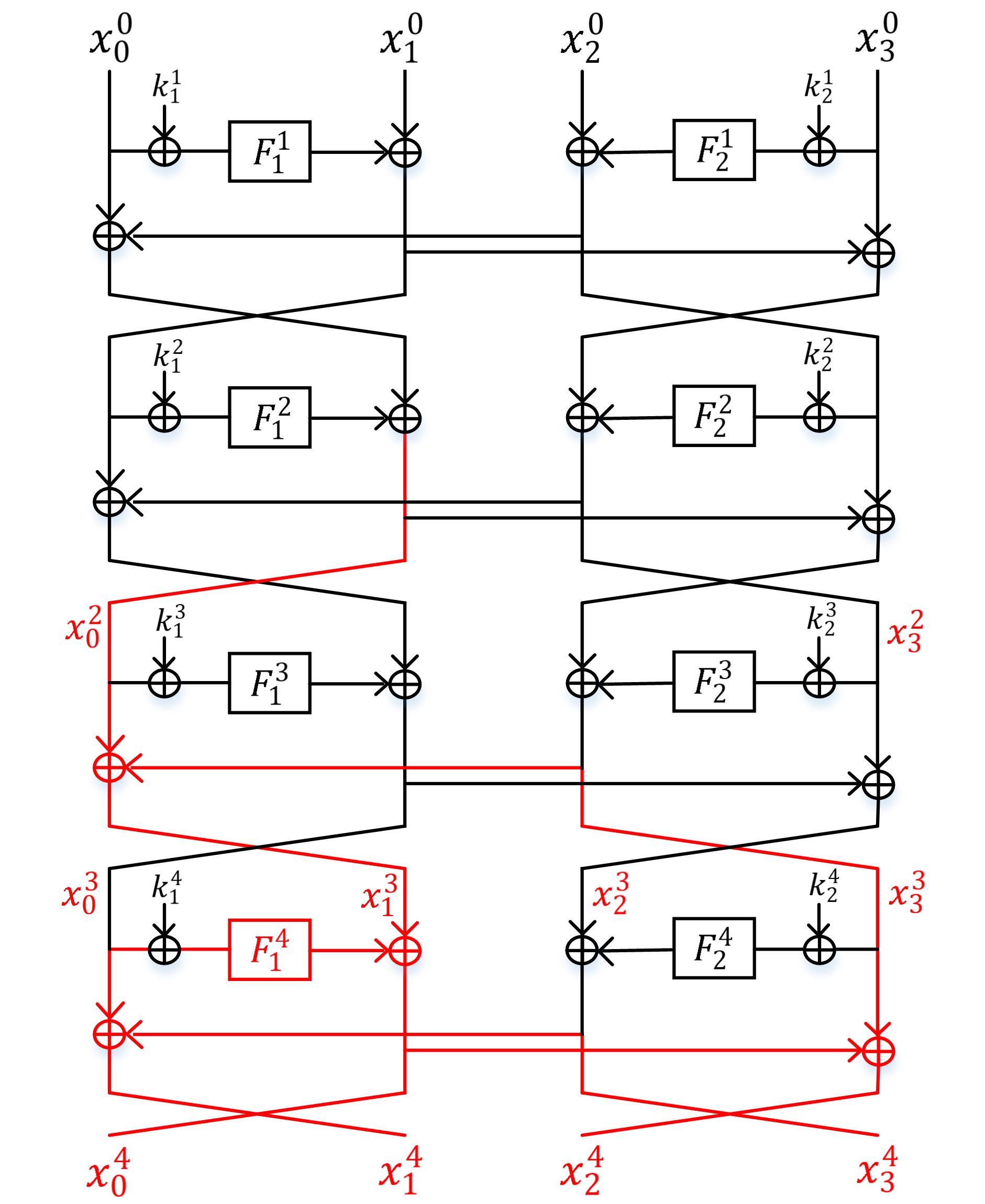}
    \caption{The 4-round FBC-KF structure, with the red line indicating the process of obtaining $x_{0}^{2}$ from the decryption perspective.}
    \label{fig10}
\end{figure}

\begin{proof}
Based on Eqs. (\ref{Eq14}) and (\ref{Eq15}), we design a low-data key-recovery attack on 4-round FBC-KF structure. This key-recovery attack is divided into the following steps.

\textbf{Step 1.} From Eq. (\ref{Eq14}), we can get $x_2^4=x_0^2\oplus F_1^4 (k_1^4\oplus x_1^4\oplus x_3^4)$, where $x_0^2=x_0^0\oplus x_2^0\oplus F_2^1 (x_3^0\oplus k_2^1)\oplus F_1^2 (x_1^0\oplus k_1^2\oplus F_1^1 (x_0^0\oplus k_1^1))$. Given any plaintext $p=(x_0, x_1, x_2, x_3)$, querying the 4-round FBC-KF encryption oracle returns the ciphertext $(y_0,y_1,y_2,y_3)$. The function $G_1(p)$ is defined as $G_1(p)=x_0\oplus x_2\oplus F_2^1 (x_3\oplus k_2^1)\oplus F_1^2 (x_1\oplus k_1^2\oplus F_1^1 (x_0\oplus k_1^1))\oplus F_1^4 (k_1^4\oplus y_1 \oplus y_3)$. The value of $G_1(p)$ corresponds to $y_2$ and can be derived from the ciphertext of $p$.

1. Query the plaintext $p_1=(x_0^0,0,x_0^0,0)$ to the encryption oracle, we can get
\begin{equation}
G_1(p_{1})=F_{2}^{1}(k_{2}^{1}) \oplus F_{1}^{2}(k_{1}^{2} \oplus F_{1}^{1}(x_{0}^{0} \oplus k_{1}^{1})) \oplus F_{1}^{4}(k_{1}^{4} \oplus x_{1}^{4} \oplus x_{3}^{4}).
\end{equation}
Then query $p_2=(x_0^0,0,0,0)$ to get a value of $G_1(p_2)$, where
\begin{equation}
\begin{aligned}
G_1(p_{2}) &= x_{0}^{0} \oplus F_{2}^{1}(k_{2}^{1}) \oplus F_{1}^{2}(k_{1}^{2} \oplus F_{1}^{1}(x_{0}^{0} \oplus k_{1}^{1}))\\
&\hspace{0.7cm} \oplus F_{1}^{4}(k_{1}^{4} \oplus (x_{1}^{4})^{\prime} \oplus (x_{3}^{4})^{\prime}).
\end{aligned}
\end{equation}
Here, \((x_{1}^{4})'\) and \((x_{3}^{4})'\)\footnote{\label{note:const}\((x_{1}^{4})'\) and \((x_{3}^{4})'\) are known constants different from \(x_{1}^{4}\) and \(x_{3}^{4}\), respectively. In the paper, notations such as \((\cdot)'\), \((\cdot)''\), \((\cdot)'^{3}\), \((\cdot)'^{4}\), and \((\cdot)'^{5}\) all denote mutually distinct constants.} denote the outputs of the second and fourth branches of the 4-round FBC-KF structure when the plaintext input is \(p_2\).
By XORing the values $G_1(p_1)$ and $G_1(p_2)$, we can obtain $G_1(p_1)\oplus  G_1(p_2)=x_0^0\oplus F_1^4 (k_1^4\oplus x_1^4\oplus x_3^4)\oplus F_1^4 (k_1^4\oplus (x_1^4)'
\oplus (x_3^4)')$. Since the values of $G_1(p_1)\oplus G_1(p_2)$, $x_{0}^{0}$, $x_{1}^{4}$, $x_{3}^{4}$, $(x_1^4)'$, and $(x_3^4)'$ are known, \textbf{the key $k_1^4$ can be calculated using Grover algorithm in $O(2^{n/2})$ time.} After obtaining $k_1^4$, the value $F_1^4 (k_1^4\oplus x_1^4\oplus x_3^4)$ can be regarded as a known parameter.

2. Query the plaintexts $p_3=(x_0^0,x_1^0,x_2^0,x_3^0)$ and $p_4=(x_0^0,x_1^0,x_2^0,(x_3^0)')$ to the encryption oracle to receive two distinct ciphertexts, where \((x_3^0)'\) is a constant different from \(x_3^0\). We have
\begin{equation}
\begin{aligned}
G_1(p_{3}) =& x_{0}^{0} \oplus x_{2}^{0} \oplus F_{2}^{1}(x_{3}^{0} \oplus k_{2}^{1}) \oplus F_{1}^{2}(x_{1}^{0} \oplus k_{1}^{2} \oplus F_{1}^{1}(x_{0}^{0}\\
&\hspace{0.2cm} \oplus k_{1}^{1})) \oplus F_{1}^{4}(k_{1}^{4} \oplus (x_{1}^{4})'' \oplus (x_{3}^{4})'')
\end{aligned}
\end{equation}
and
\begin{equation}
\begin{aligned}
G_1(p_{4}) =& x_{0}^{0} \oplus x_{2}^{0} \oplus F_{2}^{1}((x_{3}^{0})^{\prime} \oplus k_{2}^{1}) \oplus F_{1}^{2}(x_{1}^{0} \oplus k_{1}^{2} \oplus F_{1}^{1}(x_{0}^{0}\\
&\hspace{0.2cm} \oplus k_{1}^{1})) \oplus F_{1}^{4}(k_{1}^{4} \oplus (x_{1}^{4})'^{3} \oplus (x_{3}^{4})'^{3}).
\end{aligned}
\end{equation}
By XORing the values $G_1(p_{3})$ and $G_1(p_{4})$, we can obtain
\[
G_1(p_{3}) \oplus G_1(p_{4}) = F_{2}^{1}(x_{3}^{0} \oplus k_{2}^{1}) \oplus F_{2}^{1}((x_{3}^{0})^{\prime} \oplus k_{2}^{1}) \oplus C_{1},
\]
where $C_1=F_1^4 (k_1^4\oplus (x_1^4)''\oplus (x_3^4)'')\oplus F_1^4 (k_1^4\oplus (x_1^4)'^{3}\oplus (x_3^4)'^{3})$ is known. \textbf{Since $k_2^1$ is the only unknown in this expression, it can be efficiently recovered using Grover algorithm in $O(2^{n/2})$ time}.

3. Query the plaintext $p_5=(x_0^0,(x_1^0)',x_2^0,x_3^0)$ to the encryption oracle to obtain the ciphertext $((x_0^4)'^4,(x_1^4)'^4,(x_2^4)'^4,(x_3^4)'^4)$, and then compute the XOR of $G_1(p_3)$ and $G_1(p_5)$ to get
\begin{equation}\label{Eq20}
G_1(p_{3}) \oplus G_1(p_{5}) = F_{1}^{2}(\beta_{1}) \oplus F_{1}^{2}(\beta_{2}) \oplus C_{2},
\end{equation}
where $\beta_{1}=x_1^0\oplus k_1^2\oplus F_1^1 (x_0^0\oplus k_1^1)$, $\beta_{2}=(x_1^0)'\oplus k_1^2\oplus F_1^1 (x_0^0\oplus k_1^1)$, and $C_{2}=F_1^4 (k_1^4\oplus (x_1^4)''\oplus (x_3^4)'')\oplus F_1^4 (k_1^4\oplus (x_1^4)'^4\oplus (x_3^4)'^4)$. It follows that $\beta_{1} \oplus \beta_{2}=x_1^0\oplus (x_1^0)'$. By substituting this into Eq. (\ref{Eq20}), we get $G_1(p_3)\oplus G_1(p_5)=F_1^2 (\beta_{1})\oplus F_1^2 (\beta_{1} \oplus x_1^0\oplus (x_1^0)')\oplus C_2$. Since $\beta_{1}$ is the only unknown, it can be efficiently obtained using Grover algorithm in $O(2^{n/2})$ time.

4. Query the plaintext $p_6=((x_0^0)',x_1^0,x_2^0,x_3^0)$ to the encryption oracle to receive the ciphertext. Then, XOR $G_1(p_3)$ and $G_1(p_6)$ to get
\[
G_1(p_3)\oplus G_1(p_6)=F_1^2 (\beta_{1})\oplus F_1^2 (\beta_{3})\oplus C_3,
\]
where $\beta_{3}=x_1^0\oplus k_1^2\oplus F_1^1 ((x_0^0)'\oplus k_1^1)$ and $C_3=F_1^4 (k_1^4\oplus (x_1^4)''\oplus (x_3^4)'')\oplus F_1^4 (k_1^4\oplus (x_1^4)'^5\oplus (x_3^4)'^5)\oplus x_0^0\oplus (x_0^0)'$. Since $G_1(p_3)$, $G_1(p_6)$, $\beta_{1}$, and $C_3$ are known, we can obtain $\beta_{3}$ by using Grover algorithm in $O(2^{n/2})$ time. By computing the XOR of $\beta_{1}$ and  $\beta_{3}$, we get $\beta_{1}\oplus \beta_{3}=F_1^1 (x_0^0\oplus k_1^1)\oplus F_1^1 ((x_0^0)'\oplus k_1^1)$. \textbf{The key $k_1^1$ can be efficiently recovered using Grover algorithm in $O(2^{n/2})$ time. Finally, $k_1^2$ can be derived from $\beta_{1}=x_1^0\oplus k_1^2\oplus F_1^1 (x_0^0\oplus k_1^1)$.}

\textbf{Step 2.}
According to Eq. (\ref{Eq15}), we define the function $G_2(p)$ as
\begin{equation}
\begin{aligned}
G_2(p)& = x_{3} \oplus F_{2}^{2}(x_{2} \oplus k_{2}^{2} \oplus F_{2}^{1}(x_{3} \oplus k_{2}^{1})) \oplus F_{1}^{3}(x_{0} \oplus x_{2}\oplus\\
&\hspace{0.4cm}  k_{1}^{3} \oplus  F_{2}^{1}(x_{3} \oplus k_{2}^{1}) \oplus F_{1}^{2}(x_{1} \oplus k_{1}^{2} \oplus F_{1}^{1}(x_{0} \oplus k_{1}^{1}))),
\end{aligned}
\end{equation}
where $p=(x_0,x_1,x_2,x_3)$ is an arbitrary input.

1. Query the plaintexts $p_3=(x_0^0,x_1^0,x_2^0,x_3^0)$ and $p_5=(x_0^0,(x_1^0)',x_2^0,x_3^0)$ to the encryption oracle, we can easily obtain
\begin{equation}\label{Eq22}
G_2(p_{3}) \oplus G_2(p_{5})=F_{1}^{3}(\beta_{4}) \oplus F_{1}^{3}(\beta_{5}),
\end{equation}
where $\beta_{4}=x_0^0\oplus x_2^0 \oplus k_1^3\oplus F_2^1 (x_3^0\oplus k_2^1)\oplus F_1^2 (x_1^0\oplus k_1^2\oplus F_1^1 (x_0^0\oplus k_1^1))$ and $\beta_{5}=x_0^0\oplus x_2^0 \oplus k_1^3\oplus F_2^1 (x_3^0 \oplus k_2^1)\oplus F_1^2 ((x_1^0)'\oplus k_1^2\oplus F_1^1 (x_0^0\oplus k_1^1))$. It is worth noting that in Eq. (\ref{Eq22}), only $k_1^3$ remains unknown. \textbf{Therefore, $k_1^3$ can be recovered in $O(2^{n/2})$ time using Grover algorithm.}

2. \textbf{At this point, $G_2(p_3)$ contains only one unknown value $k_2^2$, which can be recovered using Grover algorithm in $O(2^{n/2})$ time}.

3. Similarly, by analyzing the outputs $x_3^2$ and $x_3^3$, \textbf{the keys $k_2^3$ and $k_2^4$ can also be recovered using Grover algorithm in $O(2^{n/2})$ time.}
\end{proof}

\textbf{Analysis.} In the Q1 model, the key-recovery attack on the 4-round FBC-KF structure requires only a constant number of plaintext-ciphertext pairs. The time complexity of Grover search for each intermediate state does not exceed $O(2^{n/2})$. Consequently, the overall attack runs in $O(2^{n/2})$ time. The data complexity is $O(1)$, and the classical memory complexity is negligible.

\subsection{Key-recovery attack on 5-round FBC-FK structure}\label{4.2}

In this section, we propose a low-data quantum key-recovery attack on 5-round FBC-FK structure in the Q1 model.

\begin{figure}
    \centering
    \includegraphics[width=0.7\linewidth]{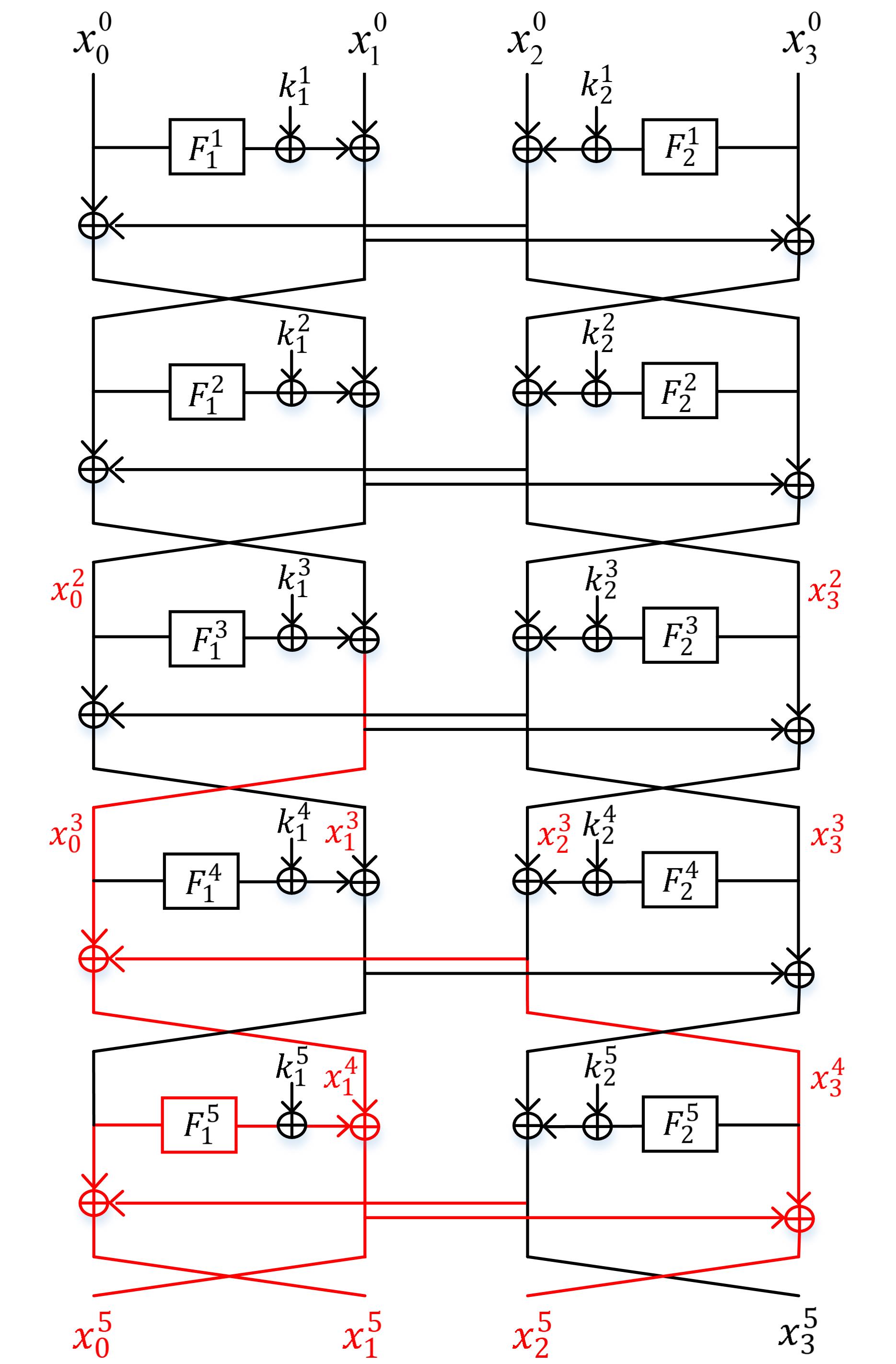}
    \caption{5-round FBC-FK structure.}
    \label{fig11}
\end{figure}
According to the 5-round FBC-FK structure shown in Fig. \ref{fig11}, we can easily calculate (red line)

\begin{equation}
\begin{aligned}
    x_{0}^{3}&=x_{1}^{4} \oplus x_{3}^{4}=k_{1}^{5} \oplus F_{1}^{5}(x_{1}^{5} \oplus x_{3}^{5}) \oplus x_{2}^{5}.
\end{aligned}
\end{equation}
Note that $F_{1}^{5}(x_{1}^{5} \oplus x_{3}^{5}) \oplus x_{2}^{5}=k_{1}^{5} \oplus x_{0}^{3}$ is a known value. Based on the encryption structure of FBC-FK, we can describe the output value $x_{0}^{3}$ as
\[
\begin{aligned}
    x_{0}^{3}= &k_{1}^{3}\oplus k_{2}^{2} \oplus x_{3}^{0} \oplus F_{2}^{2}(x_{2}^{0} \oplus k_{2}^{1} \oplus F_{2}^{1}(x_{3}^{0})) \oplus
F_{1}^{3}(k_{1}^{2} \oplus x_{0}^{0}\\
& \oplus x_{2}^{0}\oplus k_{2}^{1} \oplus F_{2}^{1}(x_{3}^{0}) \oplus F_{1}^{2}(k_{1}^{1} \oplus x_{1}^{0} \oplus F_{1}^{1}(x_{0}^{0}))).
\end{aligned}
\]
Therefore, we can express $F_{1}^{5}(x_{1}^{5} \oplus x_{3}^{5}) \oplus x_{2}^{5}$ as
\begin{equation}\label{Eq24}
\begin{aligned}
F_{1}^{5} &(x_{1}^{5} \oplus x_{3}^{5}) \oplus x_{2}^{5}\\
&=k_{1}^{5} \oplus k_{1}^{3} \oplus k_{2}^{2} \oplus x_{3}^{0} \oplus F_{2}^{2}(x_{2}^{0} \oplus k_{2}^{1} \oplus F_{2}^{1}(x_{3}^{0})) \oplus F_{1}^{3}(k_{1}^{2}\\
&\quad \oplus x_{0}^{0}
 \oplus x_{2}^{0} \oplus
k_{2}^{1} \oplus F_{2}^{1}(x_{3}^{0}) \oplus F_{1}^{2}(k_{1}^{1} \oplus x_{1}^{0} \oplus F_{1}^{1}(x_{0}^{0}))).
\end{aligned}
\end{equation}
Similarly, since $x_3^3=x_0^4\oplus x_2^4=k_2^5\oplus F_2^5 (x_0^5\oplus x_2^5)\oplus x_1^5$, we can obtain
\[
F_{2}^{5}(x_{0}^{5} \oplus x_{2}^{5}) \oplus x_{1}^{5} =k_{2}^{5} \oplus x_{3}^{3}.
\]
According to the encryption structure of FBC-FK, we can derive that
\begin{equation*}
\begin{aligned}
x_{3}^{3}& = k_{2}^{3} \oplus k_{1}^{2} \oplus x_{0}^{0} \oplus F_{1}^{2}(k_{1}^{1} \oplus x_{1}^{0} \oplus F_{1}^{1}(x_{0}^{0})) \oplus
F_{2}^{3}(k_{2}^{2}\oplus x_{3}^{0}\\
&\hspace{0.4cm} \oplus x_{1}^{0} \oplus
k_{1}^{1} \oplus F_{1}^{1}(x_{0}^{0}) \oplus F_{2}^{2}(k_{2}^{1} \oplus x_{2}^{0} \oplus F_{2}^{1}(x_{3}^{0}))).
\end{aligned}
\end{equation*}
Then, $F_{2}^{5}(x_{0}^{5} \oplus x_{2}^{5}) \oplus x_{1}^{5}$ can be written as
\begin{equation}\label{Eq25}
\begin{split}
F_{2}^{5}&(x_{0}^{5} \oplus x_{2}^{5}) \oplus x_{1}^{5} \\
&= k_{2}^{5} \oplus k_{2}^{3} \oplus k_{1}^{2} \oplus x_{0}^{0} \oplus F_{1}^{2}(k_{1}^{1} \oplus x_{1}^{0} \oplus F_{1}^{1}(x_{0}^{0}))\oplus
F_{2}^{3}(k_{2}^{2}\\
&\quad \oplus x_{3}^{0}\oplus x_{1}^{0} \oplus
k_{1}^{1} \oplus F_{1}^{1}(x_{0}^{0}) \oplus F_{2}^{2}(k_{2}^{1} \oplus x_{2}^{0} \oplus F_{2}^{1}(x_{3}^{0}))).
\end{split}
\end{equation}
By analyzing the intermediate process of the 5-round FBC-FK structure (see Fig. \ref{fig11}), we present the following lemma.

\begin{Lemma}
Given a 5-round FBC-FK encryption oracle, where $F$ is a public function and the subkey length is $n$ bits, the adversary can recover all keys in $O(2^{n/2})$ time using Grover algorithm in the Q1 model.
\end{Lemma}

\begin{proof}
Based on Eqs. (\ref{Eq24}) and (\ref{Eq25}), we design a low-data key-recovery attack on 5-round FBC-FK structure. This key-recovery attack is divided into the following steps.

\textbf{Step 1.} According to Eq. (\ref{Eq24}), for any chosen plaintext $p=(x_0,x_1,x_2,x_3)$, we define the function $H_1(p)$ as
\[
\begin{aligned}
H_1(p)&=
k_{1}^{5} \oplus k_{1}^{3} \oplus k_{2}^{2} \oplus x_{3} \oplus F_{2}^{2}(x_{2} \oplus k_{2}^{1} \oplus F_{2}^{1}(x_{3})) \oplus
F_{1}^{3}(k_{1}^{2}\\
&\quad \oplus x_{0} \oplus x_{2} \oplus k_{2}^{1} \oplus F_{2}^{1}(x_{3}) \oplus F_{1}^{2}(k_{1}^{1} \oplus x_{1} \oplus F_{1}^{1}(x_{0}))).
\end{aligned}
\]

1. Query the plaintexts $q_1=(x_0^0,F_1^1 (x_0^0),x_2^0,x_3^0)$ and $q_{2}=((x_{0}^{0})^{\prime}, F_{1}^{1}((x_{0}^{0})^{\prime}), x_{2}^{0}, x_{3}^{0})$ to the encryption oracle to receive corresponding ciphertexts $(x_0^5,x_1^5,x_2^5,x_3^5)$ and $((x_0^5)',(x_1^5)',(x_2^5)',(x_3^5)')$\footnote{$(x_i^j)'$ is a fixed constant different from $x_i^j$.}, respectively. By computing $H_1(q_1)$ and $H_1(q_2)$, we obtain
\begin{equation}
\begin{aligned}
   H_1(q_{1}) \oplus H_1(q_{2})=F_{1}^{3}(\delta_{1}) \oplus F_{1}^{3}(\delta_{2}),
\end{aligned}
\end{equation}
where $\delta_{1}=k_{1}^{2} \oplus x_{0}^{0} \oplus x_{2}^{0} \oplus k_{2}^{1} \oplus F_{2}^{1}(x_{3}^{0}) \oplus F_{1}^{2}(k_{1}^{1})$ and $\delta_{2}=k_{1}^{2} \oplus(x_{0}^{0})^{\prime} \oplus x_{2}^{0} \oplus k_{2}^{1} \oplus F_{2}^{1}(x_{3}^{0}) \oplus F_{1}^{2}(k_{1}^{1})$. We can get $\delta_{1} \oplus \delta_{2}=x_{0}^{0}\oplus (x_{0}^{0})^{\prime}$, i.e., $H_1(q_{1}) \oplus H_1(q_{2})=F_{1}^{3}(\delta_{1}) \oplus F_{1}^{3}(\delta_{1} \oplus x_{0}^{0} \oplus(x_{0}^{0})^{\prime})$. Here, $H_1(q_1)$, $H_1(q_2)$, $x_0^0$, $(x_0^0)'$, and function $F_1^3$ are known, then $\delta_{1}$ can be obtained using Grover algorithm in $O(2^{n/2})$ time.

2. Query $q_3=(x_0^0,x_1^0,x_2^0,x_3^0)$ to the encryption oracle to receive the ciphertext, and then compute the XOR of $H_1(q_1)$ and $H_1(q_3)$ to obtain
\begin{equation}\label{Eq27}
H_1(q_1)\oplus H_1(q_3)=F_1^3 (\delta_{1})\oplus F_1^3 (\delta_{3}),
\end{equation}
where $\delta_3 = k_1^2 \oplus x_0^0 \oplus x_2^0 \oplus k_2^1 \oplus F_2^1(x_3^0) \oplus F_1^2(k_1^1 \oplus x_1^0 \oplus F_1^1(x_0^0))$. Since $H_1(q_1)$, $H_1(q_3)$, $\delta_1$, and function $F_1^3$ are known, $\delta_3$ can be obtained by using Grover algorithm in $O(2^{n/2})$ time. Next, XOR $\delta_1$ and $\delta_3$ to obtain $\delta_1\oplus \delta_3=F_1^2 (k_1^1)\oplus F_1^2 (k_1^1\oplus x_1^0\oplus F_1^1 (x_0^0))$. \textbf{Only $k_1^1$ remains unknown, and its value can be calculated using Grover algorithm in $O(2^{n/2})$ time}.

3. Query the encryption oracle with $q_{4}=(x_{0}^{0}, x_{1}^{0}, x_{2}^{0},(x_{3}^{0})^{\prime})$, and compute the XOR of $H_1(q_3)$ and $H_1(q_4)$. The result satisfies
\begin{equation}\label{Eq28}
\begin{aligned}
   &H_1(q_{3}) \oplus H_1(q_{4})\\
&\hspace{0.2cm}=x_{3}^{0} \oplus (x_{3}^{0})' \oplus F_{2}^{2}(x_{2}^{0} \oplus k_{2}^{1} \oplus F_{2}^{1}(x_{3}^{0}))
\oplus F_{2}^{2}(x_{2}^{0} \oplus k_{2}^{1}\oplus \\
&\hspace{0.4cm} F_{2}^{1}((x_{3}^{0})'))\oplus F_{1}^{3}(\delta_{3})
\oplus F_{1}^{3}(\delta_{3} \oplus F_{2}^{1}(x_{3}^{0}) \oplus F_{2}^{1}((x_{3}^{0})')),
\end{aligned}
\end{equation}
where $\delta_{3}$ is calculated in Eq. (\ref{Eq27}). Thus, we can regard $x_{3}^{0} \oplus (x_{3}^{0})' \oplus F_{1}^{3}(\delta_{3}) \oplus F_{1}^{3}(\delta_{3} \oplus F_{2}^{1}(x_{3}^{0}) \oplus F_{2}^{1}((x_{3}^{0})'))
$ as a constant $C_{4}$, and then Eq. (\ref{Eq28}) reduces to
\begin{equation}
\begin{aligned}
   &H_1(q_{3}) \oplus H_1(q_{4})\\
   &\hspace{0.2cm}= F_{2}^{2}(x_{2}^{0} \oplus k_{2}^{1} \oplus F_{2}^{1}(x_{3}^{0})) \oplus F_{2}^{2}(x_{2}^{0} \oplus k_{2}^{1}\\&\hspace{0.4cm} \oplus F_{2}^{1}((x_{3}^{0})^{\prime})) \oplus C_{4}.
\end{aligned}
\end{equation}
\textbf{Here, we can use Grover algorithm to get $k_2^1$ in $O(2^{n/2})$ time}. Once $k_1^1$ and $k_2^1$ are known, \textbf{we can easily obtain $k_1^2$ from $\delta_1$.}

\textbf{Step 2.}
 Observing Eq. (\ref{Eq25}), for an arbitrary plaintext $p=(x_0,x_1,x_2,x_3)$, we assume that the function $H_2(p)=k_{2}^{5} \oplus k_{2}^{3} \oplus F_{2}^{3}(k_{2}^{2} \oplus x_{3} \oplus x_{1} \oplus k_{1}^{1} \oplus F_{1}^{1}(x_{0}) \oplus F_{2}^{2}(k_{2}^{1} \oplus x_{2} \oplus F_{2}^{1}(x_{3}))) \oplus C_{5} .$ Since some of the values in the formula have already been computed in the \textbf{step 1}, we denote it as $C_{5}$, where $C_5=k_1^2\oplus x_0\oplus F_1^2 (k_1^1\oplus x_1\oplus F_1^1 (x_0))$ is a known paramater.

1. Query the encryption oracle with $q_3=(x_0^0, x_1^0, x_2^0, x_3^0)$ and $q_5=(x_0^0, (x_1^0)', x_2^0, x_3^0)$ to obtain the ciphertexts. We then have
\begin{equation}\label{Eq30}
\begin{aligned}
H_2(q_{3}) \oplus H_2(q_{5})=F_{2}^{3}(\delta_{4}) \oplus F_{2}^{3}(\delta_{5}) \oplus C_{6} \oplus C_{7},
\end{aligned}
\end{equation}
where $\delta_{4} = k_{2}^{2} \oplus x_{3}^{0} \oplus x_{1}^{0} \oplus k_{1}^{1} \oplus F_{1}^{1}(x_{0}^{0}) \oplus F_{2}^{2}(k_{2}^{1} \oplus x_{2}^{0} \oplus F_{2}^{1}(x_{3}^{0}))
$,
$\delta_{5}=k_{2}^{2} \oplus x_{3}^{0} \oplus(x_{1}^{0})^{\prime} \oplus k_{1}^{1} \oplus F_{1}^{1}(x_{0}^{0}) \oplus F_{2}^{2}(k_{2}^{1} \oplus x_{2}^{0} \oplus F_{2}^{1}(x_{3}^{0})),
$
and the corresponding values
$C_6 = k_1^2 \oplus x_0^0 \oplus F_1^2(k_1^1 \oplus x_1^0 \oplus F_1^1(x_0^0)),
C_7 = k_1^2 \oplus x_0^0 \oplus F_1^2(k_1^1 \oplus (x_1^0)' \oplus F_1^1(x_0^0)).$
By substituting $\delta_{5}=\delta_{4}\oplus x_1^0\oplus (x_1^0)'$ into Eq. (\ref{Eq30}), the value of $\delta_{4}$ can be obtained using Grover algorithm in $O(2^{n/2})$ time. \textbf{Then $k_2^2$ can be directly computed by applying XOR operations to the terms in the expression of $\delta_{4}$.}

2. At this point, $k_1^1$, $k_2^1$, $k_1^2$ and $k_2^2$ have been computed by Grover algorithm in $O(2^{n/2})$ time. Through the encryption process involving the branches \( x_1^3 \) and \( x_2^3 \), as well as \( x_0^2 \) and \( x_3^2 \), we can sequentially derive all keys using Grover algorithm in \( O(2^{n/2}) \) time.
\end{proof}

\textbf{Analysis.}
In the key-recovery attack on the 5-round FBC-FK structure, we only need to select 5 plaintext-ciphertext pairs to recover four subkeys. The time complexity of searching for the unknown values using Grover algorithm is $O(2^{n/2})$. The time and data requirements for recovering the remaining subkeys are not expected to exceed those of the attack described above. Therefore, the time complexity of the overall attack is $O(2^{n/2})$. The data complexity is $O(1)$, and the classical memory complexity can be negligible.

\section{Conclusion}\label{sec5}

In this paper, we analyse the FBC-F, FBC-KF, and FBC-FK structures in Q1 and Q2 models respectively. In the Q2 model, we constructed 4-round quantum distinguishers for the FBC-F and FBC-KF structures, and a 6-round quantum distinguisher for the FBC-FK structure. Building on these results, we implemented quantum key-recovery attacks on the $r$-round FBC-F/KF/FK structures, which outperform generic quantum brute-force search in terms of time complexity. In addition, we propose low-data key-recovery attacks on 4-round FBC-KF and 5-round FBC-FK structures in the Q1 model. These attacks operate in the qCPA setting, maintain the advantage of low data complexity, and can recover all keys of the FBC-KF and FBC-FK structures in $O(2^{n/2})$ time.

\section*{Acknowledgments}
This research was funded by the National Natural Science Foundation of China (Grant Nos. 62171131, 62272056 and 62372048), Fujian Province Natural Science Foundation (Grant Nos. 2022J01186 and 2023J01533), Fujian Province Young and Middle-aged Teacher Education Research Project (Grant No. JAT231018)  and Open Foundation of State Key Laboratory of Networking and Switching Technology (Beijing University of Posts and Telecommunications)(SKLNST-2024-1-05).

{\appendices
\section{Details of the 4-round quantum distinguisher against FBC-KF structure}\label{A}

The construction of the 4-round distinguisher for the FBC-KF structure is similar to that for the FBC-F structure. First, consider arbitrary constants $\alpha_{0}, \alpha_{1}, x_{0}^{0}, x_{3}^{0} \in \{0,1\}^{n}$ such that $\alpha_{0}\neq \alpha_{1}$. Given a 4-round FBC-KF encryption oracle $\mathcal{O}_{5}$, we can define
\[
    \begin{aligned}
        f^{\mathcal{O}_{5}}:\{0,1\}^{n} &\rightarrow\{0,1\}^{n} \\
x &\mapsto x_{1}^{4}(\alpha_{0}) \oplus x_{3}^{4}(\alpha_{0}) \oplus x_{1}^{4}(\alpha_{1}) \oplus x_{3}^{4}(\alpha_{1}),
    \end{aligned}
\]
where \( x_i^j(\alpha_b) \) denotes the output of the \((i+1)\)-th branch in the \( j \)-th round of the FBC-KF structure when the input plaintext is \( (x_0^0, \alpha_b, x, x_3^0) \). As shown in Fig. \ref{fig12}, the indices satisfy \( i \in \{0,1,2,3\} \), \( j \in \{1,2,3,4\} \), and \( b \in \{0,1\} \). From the 4-round FBC-KF structure, we obtain
\begin{equation}
    \begin{aligned}
f^{\mathcal{O}_{5}}(x)
& =x_{1}^{4}(\alpha_{0}) \oplus x_{3}^{4}(\alpha_{0}) \oplus x_{1}^{4}(\alpha_{1}) \oplus x_{3}^{4}(\alpha_{1})\\
& =x_{0}^{3}(\alpha_{0}) \oplus x_{0}^{3}(\alpha_{1}) .
\end{aligned}
\end{equation}
Next, we have the following lemma.

\begin{figure}
    \centering
    \includegraphics[width=0.7\linewidth]{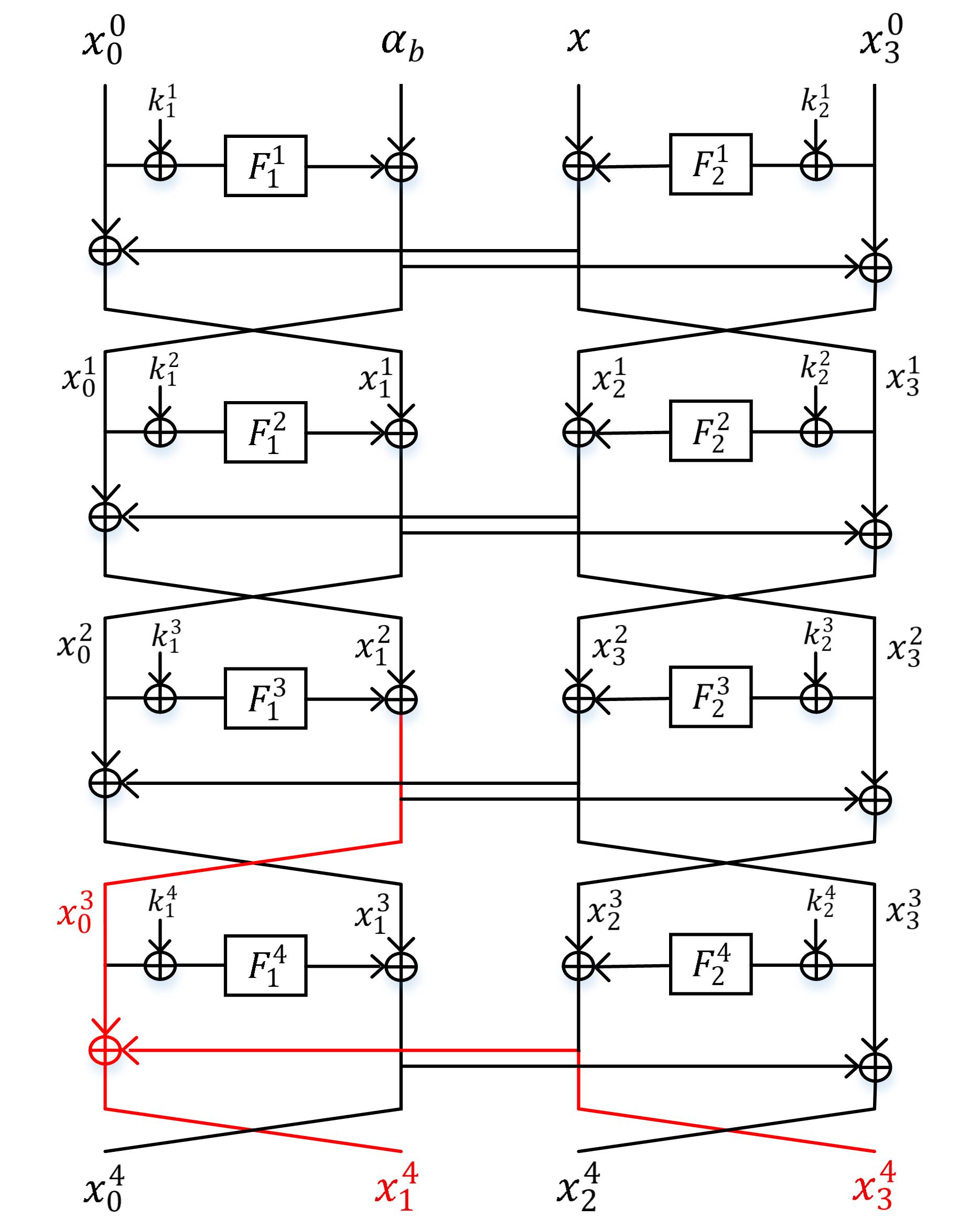}
    \caption{The 4-round FBC-KF structure, with the red line indicating the process of obtaining $x_{0}^{3}$ from the decryption perspective.}
    \label{fig12}
\end{figure}

\begin{Lemma}
If $\mathcal{O}_5$ is the encryption algorithm of $4$-round FBC-KF structure, for any $x \in \{0,1\}^n$, the function $f^{\mathcal{O}_5}$ satisfies
\[
\begin{aligned}
f^{\mathcal{O}_5}(x)=&f^{\mathcal{O}_5}(x \oplus F_{1}^{2}(\alpha_{0}\oplus k_1^2 \oplus F_{1}^{1}(x_{0}^{0}\oplus k_1^1) ) \oplus F_{1}^{2}(\alpha_{1}\\
&\hspace{0.1cm}\oplus k_1^2 \oplus F_{1}^{1}(x_{0}^{0}\oplus k_1^1))).
\end{aligned} \]
That is, $f^{\mathcal{O}_5}$ has the period $s=F_{1}^{2}(\alpha_{0}\oplus k_1^2 \oplus F_{1}^{1}(x_{0}^{0}\oplus k_1^1) ) \oplus F_{1}^{2}(\alpha_{1}\oplus k_1^2 \oplus F_{1}^{1}(x_{0}^{0}\oplus k_1^1))$.
\end{Lemma}

\begin{proof}
 By choosing the plaintext $(x_{0}^{0}, \alpha_{b}, x, x_{3}^{0})$ as the input of the FBC-KF structure, the output of the first round is $(x_{0}^{1}, x_{1}^{1}, x_{2}^{1}, x_{3}^{1})$, where

\begin{equation}\label{Eq32}
\begin{aligned}
x_{0}^{1}&=F_{1}^{1}(x_{0}^{0}\oplus k_{1}^{1}) \oplus \alpha_{b},\\
x_{1}^{1}&=x_{0}^{0} \oplus x \oplus F_{2}^{1}(x_{3}^{0}\oplus k_{2}^{1}),\\
x_{2}^{1}&=x_{3}^{0}\oplus \alpha_{b} \oplus F_{1}^{1}(x_{0}^{0}\oplus k_{1}^{1}) ,\\
x_{3}^{1}&=x \oplus F_{2}^{1}(x_{3}^{0}\oplus k_{2}^{1}).
\end{aligned}
\end{equation}
Similarly, the second round output $(x_{0}^{2}, x_{1}^{2}, x_{2}^{2}, x_{3}^{2})$ can be expressed as
\begin{equation}\label{Eq33}
  \begin{aligned}
x_{0}^{2}&=x \oplus x_{0}^{0} \oplus F_{2}^{1}(x_{3}^{0}\oplus k_{2}^{1}) \oplus F_{1}^{2}(\alpha_{b}\oplus  k_{1}^{2}\\
&\hspace{0.5cm} \oplus F_{1}^{1}(x_{0}^{0}\oplus k_{1}^{1})), \\
x_{1}^{2}&=x_{3}^{0} \oplus F_{2}^{2}(x \oplus k_{2}^{2}\oplus  F_{2}^{1}(x_{3}^{0}\oplus k_{2}^{1})), \\
x_{2}^{2}&=x_{0}^{0} \oplus F_{1}^{2}( k_{1}^{2}\oplus \alpha_{b} \oplus F_{1}^{1}(k_{1}^{1}\oplus x_{0}^{0})), \\
x_{3}^{2}&=x_{3}^{0}\oplus \alpha_{b}  \oplus F_{1}^{1}(k_{1}^{1}\oplus x_{0}^{0}) \oplus F_{2}^{2}(x\oplus k_{2}^{2}\\
&\hspace{0.5cm}\oplus F_{2}^{1}(x_{3}^{0}\oplus k_{2}^{1})).
\end{aligned}
\end{equation}
As shown by the red line in Fig. \ref{fig12}, using the ciphertexts \( x_{1}^4 \) and \( x_{3}^4 \) obtained from the encryption oracle \(\mathcal{O}_5\) of the 4-round FBC-KF structure, we can derive \( x_{0}^3 \) from the relation
\[
x_{1}^4 \oplus x_{3}^4 = x_{0}^3.
\]
Meanwhile, based on the FBC-KF encryption structure, \( x_0^3 \) can also be expressed as
\begin{equation}\label{Eq34}
\begin{aligned}
x_{0}^{3}=
F_{1}^{3}(x_{0}^{2}\oplus k_{1}^{3}) \oplus x_{1}^{2}.
\end{aligned}
\end{equation}
Therefore, it can be seen from Eqs. (\ref{Eq33}) and (\ref{Eq34}) that
\[
\begin{aligned}
x_{0}^{3}=&x_{3}^{0} \oplus F_{2}^{2}(x\oplus k_{2}^{2} \oplus F_{2}^{1}(x_{3}^{0}\oplus k_{2}^{1})) \oplus F_{1}^{3}( x \oplus x_{0}^{0}\oplus k_{1}^{3}  \\
&\hspace{0.2cm}\oplus F_{2}^{1}(x_{3}^{0}\oplus k_{2}^{1})\oplus  F_{1}^{2}( \alpha_{b}\oplus  k_{1}^{2}\oplus F_{1}^{1}(x_{0}^{0}\oplus k_{1}^{1}))).
\end{aligned}\]
Hence, the function $f^{\mathcal{O}_5}$ can be described as
\[
    \begin{aligned}
f^{\mathcal{O}_5}(x)&=x_{0}^{3}(\alpha_{0}) \oplus x_{0}^{3}(\alpha_{1}) \\
&=F_{1}^{3}( x \oplus x_{0}^{0}\oplus k_{1}^{3} \oplus F_{2}^{1}(x_{3}^{0}\oplus k_{2}^{1})\oplus  F_{1}^{2}( \alpha_{0}\oplus  k_{1}^{2} \oplus\\
&\hspace{0.2cm} F_{1}^{1}(x_{0}^{0}\oplus k_{1}^{1}))) \oplus F_{1}^{3}( x \oplus x_{0}^{0}
\oplus k_{1}^{3}\oplus F_{2}^{1}(x_{3}^{0}\oplus k_{2}^{1})\\
&\hspace{0.2cm}\oplus  F_{1}^{2}( \alpha_{1}\oplus  k_{1}^{2}\oplus F_{1}^{1}(x_{0}^{0}\oplus k_{1}^{1}))).
\end{aligned}
\]
The function $f^{\mathcal{O}_5}$ has a period $s$ since it satisfies

\[
\begin{aligned}
&f^{\mathcal{O}_5}(x \oplus F_{1}^{2}(\alpha_{0}\oplus k_1^2 \oplus F_{1}^{1}(x_{0}^{0}\oplus k_1^1))\\&\hspace{1.2cm} \oplus F_{1}^{2}(\alpha_{1}\oplus k_1^2 \oplus F_{1}^{1}(x_{0}^{0}\oplus k_1^1)))\\
&\hspace{0.2cm}=F_{1}^{3}( x \oplus \textcolor{red}{F_{1}^{2}(\alpha_{0}\oplus k_1^2 \oplus F_{1}^{1}(x_{0}^{0}\oplus k_1^1))} \oplus F_{1}^{2}(\alpha_{1} \oplus  k_1^2\\
&\hspace{0.6cm} \oplus F_{1}^{1}(x_{0}^{0}\oplus k_1^1)) \oplus x_{0}^{0}\oplus k_{1}^{3} \oplus F_{2}^{1}(x_{3}^{0}\oplus k_{2}^{1}) \oplus \textcolor{red}{F_{1}^{2}( \alpha_{0}} \\
&\hspace{0.68cm}
\textcolor{red}{\oplus  k_{1}^{2}\oplus F_{1}^{1}(x_{0}^{0}\oplus k_{1}^{1}))}) \oplus F_{1}^{3}( x \oplus F_{1}^{2}(\alpha_{0}\oplus k_1^2 \oplus F_{1}^{1}(\\
&\hspace{0.68cm} x_{0}^{0} \oplus k_1^1))\oplus \textcolor{blue}{F_{1}^{2}(\alpha_{1}\oplus k_1^2 \oplus F_{1}^{1}(x_{0}^{0}\oplus k_1^1))}\oplus x_{0}^{0} \oplus k_{1}^{3}\\
&\hspace{0.68cm} \oplus F_{2}^{1}(x_{3}^{0}\oplus k_{2}^{1})\oplus  \textcolor{blue}{F_{1}^{2}( \alpha_{1} \oplus  k_{1}^{2}\oplus F_{1}^{1}(x_{0}^{0}\oplus k_{1}^{1}))})\\
&\hspace{0.2cm}=F_{1}^{3}( x \oplus F_{1}^{2}(\alpha_{1}\oplus k_1^2 \oplus F_{1}^{1}(x_{0}^{0}\oplus k_1^1)) \oplus x_{0}^{0}\oplus k_{1}^{3} \oplus\\
&\hspace{0.68cm} F_{2}^{1}(x_{3}^{0}\oplus k_{2}^{1})) \oplus F_{1}^{3}( x \oplus F_{1}^{2}(\alpha_{0}\oplus k_1^2 \oplus F_{1}^{1}(x_{0}^{0} \oplus k_1^1))\\
&\hspace{0.68cm} \oplus x_{0}^{0} \oplus k_{1}^{3}\oplus F_{2}^{1}(x_{3}^{0}\oplus k_{2}^{1})) \\
&\hspace{0.2cm}=f^{\mathcal{O}_5}(x),
\end{aligned}
\]
where $s=F_{1}^{2}(\alpha_{0}\oplus k_1^2 \oplus F_{1}^{1}(x_{0}^{0}\oplus k_1^1)) \oplus F_{1}^{2}(\alpha_{1}\oplus k_1^2 \oplus F_{1}^{1}(x_{0}^{0}\oplus k_1^1)).$ The proof of the above lemma holds.
\end{proof}

In summary, using the function $f^{\mathcal{O}_5}$, we can construct a 4-round distinguisher for the FBC-KF structure. The period $s$ can then be efficiently recovered by the Simon algorithm in $O(n)$ time.
}

\begin{IEEEbiography}[{\includegraphics[width=1in,height=1.2in,clip,keepaspectratio]{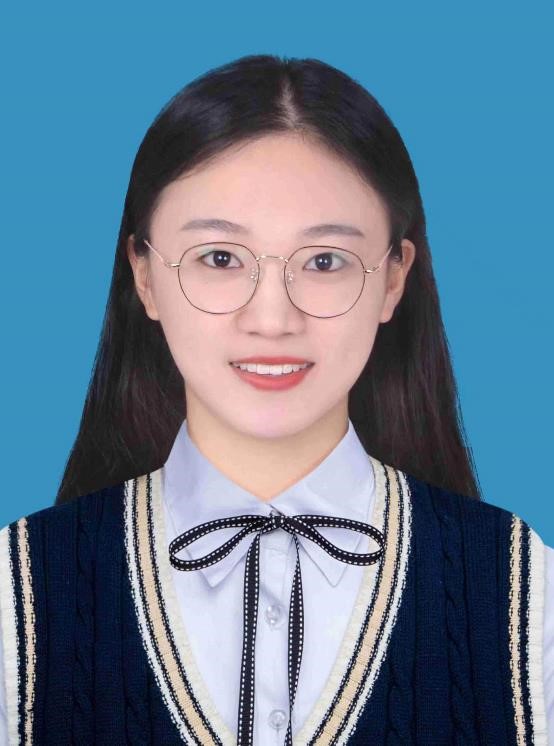}}]{Yan-Ying Zhu}
was born in 1998. She received the M.S. degree in Mathematics from Shandong University of Science and Technology, Qingdao, China, in 2023. She is currently pursuing the Ph.D. degree in Cyberspace Security, Fujian Normal University, Fuzhou, China. Her current research interests is quantum cryptography and quantum computation.\end{IEEEbiography}

\begin{IEEEbiography}[{\includegraphics[width=1in,height=1.3in,clip,keepaspectratio]{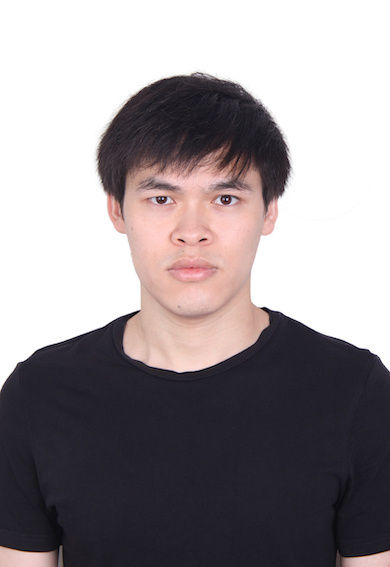}}]{Bin-Bin Cai}
received the B.S. degree in computer science and M.S. degree in computer application both from Fujian Normal University, Fuzhou, China, in 2014 and 2017, respectively, and Ph.D. degree in cyberspace security from Beijing University of Posts and Telcommunications, Beijing, China, in 2023. He is currently a lecturer at the College of Computer and Cyber Security, Fujian Normal University, China. His research interests include quantum cryptography, quantum cryptanalysis and quantum computing.\end{IEEEbiography}

\begin{IEEEbiography}[{\includegraphics[width=1in,height=1.25in,clip,keepaspectratio]{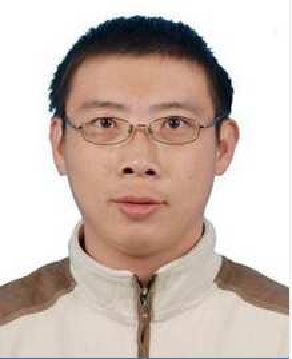}}]{Fei Gao}
received the B.E. degree in communication engineering and the Ph.D. degree in cryptography from the Beijing University of Posts and Telecommunications (BUPT), Beijing, China, in 2002 and 2007, respectively. He is currently the Director of the State Key Laboratory of Networking and Switching Technology, Network Security Research Center (NSRC), BUPT, and also with the Peng Cheng Laboratory, Center for Quantum Computing, Shenzhen, China. He is also working on the practical quantum cryptographic protocols, quantum nonlocality, and quantum algorithms. Prof. Gao is a member of the Chinese Association for Cryptologic Research (CACR).
\end{IEEEbiography}

\begin{IEEEbiography}[{\includegraphics[width=1in,height=1.25in,clip,keepaspectratio]{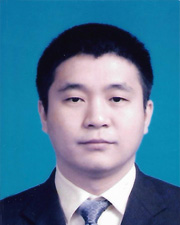}}]{Song Lin}
received the B.S. degree in computer science and M.S. degree in mathematics both from Fujian Normal University, Fuzhou, China, in 1999 and 2005, respectively, and Ph.D. degree in cryptography from Beijing University of Posts and Telcommunications, Beijing, China, in 2009. He is currently a Professor at the College of Computer and Cyber Security, Fujian Normal University, China. His research interests include quantum cryptography, information security, and quantum computing. Prof. Lin is a member of the Chinese Association for Cryptologic Research.\end{IEEEbiography}


\begin{thebibliography}{22}
\bibitem{Su2025}
Su J J, Fan J C, Wu S Y, et al. Topology-driven quantum architecture search framework. Science China Information Sciences, 2025, 68(8): 180507.
\bibitem{Wu2024}
Wu S Y, Li R Z, Song Y Q, et al. Quantum assisted hierarchical fuzzy neural network for image classification. IEEE Transactions on Fuzzy Systems, 2024.
\bibitem{Wu2025}
Wu S Y, Song Y Q, Li R Z, et al. Resource-Efficient Adaptive Variational Quantum Algorithm for Combinatorial Optimization Problems. Advanced Quantum Technologies, 2025: 2400484.
\bibitem{Kuwakado2010}
Kuwakado H, Morii M. Quantum distinguisher between the 3-round Feistel cipher and the random permutation. In: Proceedings of International Symposium on Information Theory. 2010, 2682-2685.
\bibitem{Kuwakado20102}
Kuwakado H, Morii M. Security on the quantum-type Even-Mansour cipher. In: Proceedings of International Symposium on Information Theory and Its Applications. 2012, 312-316.
\bibitem{Kaplan2016}
Kaplan M, Leurent G, Leverrier A, et al. Breaking symmetric cryptosystems using quantum period finding. In: Advances in Cryptology-CRYPTO 2016. Berlin: Springer-Verlag. 2016, 207¨C237.
\bibitem{Shor1994}
Shor P W. Algorithms for quantum computation: discrete logarithms and factoring. In: 35th annual symposium on foundations of computer science. IEEE, 1994, 124-134.
\bibitem{Shor1997}
Shor P W. Polynomial-time algorithms for prime factorization and discrete logarithms on a quantum computer. SIAM Journal on Computing, vol. 26, 1997, pp. 1484¨C1509.
\bibitem{RSA1978}
Rivest R L, Shamir A, Adleman L. A method for obtaining digital signatures and public-key cryptosystems. Communications of the ACM, 1978, 21(2): 120-126.
\bibitem{Grover1996}
Grover L K. A fast quantum mechanical algorithm for database search. In: Proceedings of the twenty-eighth annual ACM symposium on Theory of computing. 1996, 212-219.
\bibitem{Simon1997}
Simon D R. On the power of quantum computation. SIAM journal on computing, 1997, 26(5): 1474-1483.
\bibitem{Luo2019}
Luo Y Y, Yan H L, Wang L, et al. Study on block cipher structures against simon's quantum algorithm. Journal of Cryptologic Research, 2019, 6(5): 561¨C573.
\bibitem{You2020}
You Q D, Qian X, Zhou X, et al. Research on quantum cryptanalysis on SMS4-like structure and NBC algorithm. Journal of Cryptologic Research, 2020, 7(6): 864¨C874.
\bibitem{Qian2021}
Qian X, You Q D, Zhou X, et al. Quantum attack on MARS-like Feistel schemes. Journal of Cryptologic Research, 2021, 8(3): 417¨C431.
\bibitem{Li2022}
Li Y J, Yi Z H, Wang Z, et al. Quantum cryptanalysis of lightweight cipher TWINE-128. Journal of Cryptologic Research, 2022, 9(4): 633¨C643.
\bibitem{Li2021}
Li Y J, Lin H, Yi Z H, et al. Quantum cryptanalysis of MIBS. Journal of Cryptologic Research, 2021, 8(6): 989¨C998.
\bibitem{Leander2017}
Leander G, May A. Grover meets simon - quantumly attacking the FX-construction. In: Advances in Cryptology -ASIACRYPT 2017, Part II. Berlin: Springer. 2017, 10625: 161¨C178.
\bibitem{Kilian1996}
Kilian J, Rogaway P. How to protect DES against exhaustive key search. In: Advances in Cryptology-CRYPTO 1996. Berlin: Springer. 1996, 1109: 252¨C267.
\bibitem{Onions2001}
Onions P. On the strength of simply-iterated Feistel ciphers with whitening keys. In: Cryptographers' Track at the RSA Conference. Berlin: Springer. 2001, 63-69.
\bibitem{Dong2018}
Dong X Y, Wang X Y. Quantum key-recovery attack on Feistel structures. Science China Information Sciences, 2018, 61(10): 102501.
\bibitem{Ito20192}
Ito G, Hosoyamada A, Matsumoto R, et al. Quantum chosen-ciphertext attacks against Feistel ciphers. In: Topics in Cryptology¨CCT-RSA 2019. Cham: Springer. 2019, 11405: 391¨C411.
\bibitem{Canale2022}
Canale F, Leander G, Stennes L. Simon¡¯s algorithm and symmetric crypto: Generalizations and automatized applications. In: Annual International Cryptology Conference. 2022, 779-808.
\bibitem{Dong2019}
Dong X Y, Li Z, Wang X Y. Quantum cryptanalysis on some generalized Feistel schemes. Science China Information Sciences, 2019, 62(2), 22501.
\bibitem{Ni2019}
Ni B Y, Ito G, Dong X Y, et al. Quantum attacks against type-1 generalized Feistel ciphers and applications to CAST-256. In: Progress in Cryptology ¨C INDOCRYPT 2019, Lecture Notes in Computer Science, vol. 11898. Cham: Springer. 2019: 433¨C455.
\bibitem{Sun2023}
Sun H W, Cai B B, Qin S J, et al. Quantum Attacks on Type-1 Generalized Feistel Schemes. Advanced Quantum Technologies, 2023, 6(10): 2300155.
\bibitem{S2020}
Hod\v{z}i\'{c} S, Knudsen Ramkilde L, Brasen Kidmose A. On quantum distinguishers for type-3 generalized Feistel network based on separability. In: International Conference on Post-Quantum Cryptography. Cham: Springer, 2020, 461-480.
\bibitem{Zhang2023}
Zhang Z Y, Wu W L, Sui H, et al. Quantum attacks on type-3 generalized Feistel scheme and unbalanced Feistel scheme with expanding functions. Chinese Journal of Electronics, 2023, 32(2): 209-216.
\bibitem{Cui2021}
Cui J Y, Guo J S, Ding S Z. Applications of Simon¡¯s algorithm in quantum attacks on Feistel variants. Quantum Information Processing, 2021, 20: 1-50.
\bibitem{Bo2021}
Yu B, Sun B, Liu G Q, et al. Quantum cryptanalysis on some generalized unbalanced Feistel networks. Journal of Cryptologic Research, 2021, 8(6): 960-973.
\bibitem{Kun2022}
Zou K, Dong X F, Zhang F Z. Improved quantum attack on several generalized unbalanced Feistel structures. In: 2022 International Conference on Networks, Communications and Information Technology (CNCIT). IEEE, 2022, 113-121.
\bibitem{Bonnetain2021}
Bonnetain X, Leurent G, Naya-Plasencia M, et al. Quantum linearization attacks. In: Advances in Cryptology¨CASIACRYPT 2021, Part I. Singapore: Springer. 2021, 422¨C452.
\bibitem{Zhou2021}
Zhou B M, Yuan Z. Quantum key-recovery attack on Feistel constructions: Bernstein¨CVazirani meet Grover algorithm. Quantum Information Processing, 2021, 20(10): 330.
\bibitem{Bon20192}
Bonnetain X, Naya-Plasencia M, Schrottenloher A. On quantum slide attacks. In: International Conference on Selected Areas in Cryptography. Cham: Springer. 2019, 492-519.
\bibitem{Tan2022}
Tan J W, Xiao L G, Qiu D W, et al. Distributed quantum algorithm for Simon's problem. Physical Review A, 2022, 106(3): 032417.
\bibitem{Zhou2023}
Zhou X, Qiu D W, Luo L. Distributed Bernstein¨CVazirani algorithm. Physica A: Statistical Mechanics and its Applications, 2023, 629: 129209.
\bibitem{Kaplan20162}
Kaplan M, Leurent G, Leverrier A, et al. Quantum differential and linear cryptanalysis. IACR Transactions on Symmetric Cryptology, 2016(1): 71¨C94.
\bibitem{Bonnetain2019}
Bonnetain X, Hosoyamada A, Naya-Plasencia M, et al. Quantum attacks without superposition queries: the offline Simon¡¯s algorithm. In: International Conference on the Theory and Application of Cryptology and Information Security. Cham: Springer. 2019, 552-583.
\bibitem{Rahman2020}
Rahman M., Paul G. Quantum attacks on HCTR and its variants. IEEE Transactions on Quantum Engineering, 2020, 1: 1¨C8.
\bibitem{Bonnet2020}
Bonnetain X, Jaques S. Quantum Period Finding against Symmetric Primitives in Practice. IACR Transactions on Cryptographic Hardware and Embedded Systems. 2022(1), 1-27.
\bibitem{Yu2023}
Yu B, Shi T R, Dong X Y, et al. Quantum Attacks: A View of Data Complexity on Offline Simon¡¯s Algorithm. In: International Conference on Information Security and Cryptology. Singapore: Springer. 2023: 329-342.
\bibitem{Daiza2022}
Daiza T, Yoneyama K. Quantum key recovery attacks on 3-round Feistel-2 structure without quantum encryption oracles. In: International Workshop on Security. Cham: Springer. 2022: 128-144.
\bibitem{Xu2024}
Xu Y S, Luo Y Y, Dong X Y, et al. Low-data Quantum Key-recovery Attack on Block Cipher Structures. Journal of Software (in Chinese), 2025, 36(7): 3321-3338.
\bibitem{KFeng2019}
Feng X T, Zeng X Y, Zhang F, et al. On the lightweight block cipher FBC. Journal of Cryptologic Research, 2019, 6(6): 768-785.
\bibitem{WangP2020}
Wang P, Tian S P, Sun Z W, et al. Quantum algorithms for hash preimage attacks. Quantum Engineering, 2020, 2(2): e36.
\bibitem{WangR2021}
Wang R. Comparing grover¡¯s quantum search algorithm with classical algorithm on solving satisfiability problem. In: 2021 IEEE Integrated STEM Education Conference (ISEC). 2021, 204-204.
\bibitem{Kain2021}
Kain B. Searching a quantum database with Grover's search algorithm. American journal of physics, 2021, 89(6): 618-626.
\end{thebibliography}
\end{document}